\documentclass[draft, onecolumn, twoside, 12pt]{IEEEtran}
\usepackage{amsmath, amssymb,graphicx,psfrag,cite}
\usepackage{srcltx,bm}
\usepackage{color}
\renewcommand{\QED}{\QEDopen}

\newcommand{\mynote}[1]{\color{black} #1\color{black}\null}

\newcommand{\bl}[1] {\text{\boldmath ${\lambda}$}_{#1}}
\newcommand{\EX}[1] {{\mathbb{E}}\left\{{#1}\right\}}
\newcommand{\EXs}[2] {{\mathbb{E}}_{{#1}}\!\!\left\{{#2}\right\}}

\newcommand{\tr}[1] {{\text{tr}}\left\{#1\right\}}

\def\Do{{\mathcal{D}}_{\text{ord}}}
\def\nt{n_{\text{T}}}
\def\nr{n_{\text{R}}}
\def\nmin{n_{\text{min}}}
\def\Nt{{N_{\text{T}}}}
\def\Nr{{N_{\text{R}}}}

\def\Ni{N_{\text{I}}}

\def\Ni{N_{\text{I}}}

\def\Pt{P}

\def\H{{\bf H}}
\def\y{{\bf y}}
\def\x{{\bf x}}
\def\n{{\bf n}}
\def\r{{\varrho}}
\def\I{{\bf I}}
\newcommand{\td}[1] {\tilde{#1}}
\newcommand{\mb}[1] {{\mathbf{#1}}}
\def\teq{\triangleq}
\newcommand{\diag}[1] {\text{diag}\left({#1}\right)}
\newcommand{\ith}[1]    {{#1}^{\underline{\text{th}}}}

\newtheorem{theorem}{Theorem}
\newtheorem{definition}{Definition}
\newtheorem{corollary}{Corollary}
\newtheorem{lemma}{Lemma}

\newcommand{\ac}[1]{#1}   
\newcommand{\acp}[1]{#1}   

\begin{document}




\title{MIMO Networks: the Effects of Interference
}
\author{
    Marco~Chiani,~\IEEEmembership{Senior~Member,~IEEE}, 
    Moe~Z.~Win,~\IEEEmembership{Fellow,~IEEE}, and 
    Hyundong Shin,~\IEEEmembership{Member,~IEEE}.\\
    \vspace{2cm} 
    \underline{Corresponding Address:}\\
    Marco Chiani\\
    DEIS, University of Bologna\\
    V.le Risorgimento 2, 40136 Bologna, ITALY\\
    \bigskip
    Tel: +39-0512093084 \qquad Fax: +39-0512093540\\
    e-mail: {\tt marco.chiani@unibo.it}
\thanks{This research was supported, in part, by
         the European Commission in the scope of the FP7 project CoExisting Short Range Radio
        by Advanced Ultra-WideBand Radio Technology (EUWB),
       the National Science Foundation under Grants ECCS-0636519 and ECCS-0901034,
       the Office of Naval Research Presidential Early Career Award for Scientists and Engineers (PECASE) N00014-09-1-0435,
    the MIT Institute for Soldier Nanotechnologies,
       the Korea Science and Engineering Foundation (KOSEF) grant funded by the Korea government (MOST) (No. R01-2007-000-11202-0),
and the Basic Science Research Program through the National Research Foundation of Korea (NRF) funded by the Ministry of Education, Science and Technology.}
\thanks{M.\ Chiani is with WiLab/DEIS,
        University of Bologna,
        V.le Risorgimento 2,
        40136 Bologna, ITALY
        (e-mail: {\tt marco.chiani@unibo.it}).}
\thanks{M.\ Z.\ Win is with
        the Laboratory for Information and Decision
        Systems (LIDS), Massachusetts Institute of Technology,
        Room 32-D666, 77 Massachusetts Avenue, Cambridge, MA 02139
        USA (e-mail: {\tt{moewin@mit.edu}}).}
\thanks{H.\ Shin is with the
        School of Electronics and Information
        Kyung Hee University
        Yongin, Kyungki 446-701, Korea
        (e-mail: {\tt{hshin@khu.ac.kr}}). }
}
\markboth
    {Subm. to IEEE Trans. on Inf. Th.}
    {Chiani, Win, Shin: MIMO networks.}
%
\maketitle

\date{}
\thispagestyle{empty}

\newpage
\setcounter{page}{1}

\begin{abstract}
\ac{MIMO} systems  are being considered as one of the key enabling technologies for future wireless
networks. 
However, the decrease in capacity due to the presence of
interferers in \ac{MIMO} networks is not well understood.
In this paper, we develop an analytical framework to characterize
the capacity of \ac{MIMO} communication systems in the presence of
multiple \ac{MIMO} co-channel interferers and noise. We consider
the situation in which transmitters have no channel state information 
 and all links undergo Rayleigh fading.
We first generalize the determinant representation of
hypergeometric functions with matrix arguments to the case when
the argument matrices have eigenvalues of arbitrary
multiplicity.
This enables the derivation of \mynote{
the distribution of the eigenvalues of Gaussian quadratic forms and Wishart matrices with arbitrary correlation, with application to both single-user and multiuser MIMO systems. In particular, we derive} 
the ergodic mutual information for \ac{MIMO} systems in the presence of multiple \ac{MIMO}
interferers. Our analysis is valid for any number of interferers,
each with arbitrary number of antennas having possibly unequal
power levels.
This framework, therefore, accommodates the study of distributed
\ac{MIMO} systems and accounts for different spatial positions of the
\ac{MIMO} interferers.
\end{abstract}


\section{Introduction}\label{sec:intro}

The use of multiple transmitting and receiving antennas can
provide high spectral efficiency and link reliability for
point-to-point communication in fading environments \cite{Wint:87,
WintSalGit:94}. The analysis of capacity for \ac{MIMO} channels in
\cite{Fos:96 
} suggested practical receiver structures
to obtain such spectral efficiency. Since then, many studies have
been devoted to the 
 analysis  of \ac{MIMO} systems,
starting from \mynote{the ergodic \cite{Tel:99} and outage \cite{Chi:02} capacity for } 
uncorrelated fading to the case where correlation is
present at one of the two sides (either at the transmitter or at the
receiver) or at both sides \cite{ChiWinZan:J03, SmiRoySha:J03,
ShiWinLeeChi:J05}. The effect of time correlation is studied in
\cite{
GioSmiShaChi:J03}.

Only a few papers, by using simulation or approximations, have
studied the capacity of \ac{MIMO} systems in the
presence of cochannel interference. 
In particular, a simulation study is presented in
\cite{CatDriGre:00} for cellular systems, assuming up to $3$
transmit and $3$ receive antennas. The simulations \mynote{showed} that
cochannel interference can seriously degrade the overall capacity
when MIMO links are used in cellular networks.
In \cite{BluWintSol:02, Blu:03} it is studied whether, in a MIMO
multiuser scenario, it is always convenient to use all
transmitting antennas.
It was found that for some values of \ac{SNR} and \ac{SIR},
allocating all power into a single transmitting antenna, rather
than dividing the power equally among independent streams from the 
different antennas, would lead to a higher \mynote{overall system mutual information}. 
The study in \cite{BluWintSol:02,Blu:03} adopts simulation to evaluate the capacity of MIMO systems
in the presence of cochannel interference, and the difficulties in
the evaluations limited the results to a scenario with two MIMO
users employing at most two antenna elements.
In \cite{MouSimSen:03} the replica method is used to obtain
approximate moments of the capacity for MIMO systems with large
number of antenna elements including the presence of interference.
The approximation requires iterative numerical methods to solve a
system of non-linear equations, and its accuracy has to be verified
by computer simulations. A multiuser MIMO system with specific receiver structures is analyzed for the
interference-limited case in \cite{DaiPoo:03,DaiMolPoo:04}.

The MIMO capacity at high and low \ac{SNR} for interference-limited scenarios is addressed in \cite{LozTulVer:03,LozTulVer:05}.
A worst-case analysis for MIMO capacity with \ac{CSI} at the transmitter and at the receiver, conditioned on the channel matrix, can be found in \cite{JorBoc:04}.
Asymptotic results for the Rician channel in the presence of interference can be found in \cite{TarRie:07}.

In this paper, we develop an analytical framework to analyze the
ergodic capacity of \ac{MIMO} systems in the presence of multiple
\ac{MIMO} cochannel interferers and \ac{AWGN}. Throughout the
paper we consider rich scattering environments in which
transmitters have no \ac{CSI}, the receiver has perfect \ac{CSI},
and all links undergo frequency flat Rayleigh fading.
The key contributions of the paper are as follows:
\begin{itemize}
\item Generalization of the determinant representation of
hypergeometric functions with matrix arguments
 to the case where matrices in the arguments
have eigenvalues with arbitrary multiplicity.
\item
Derivation, using the generalized representation, of the joint
\ac{p.d.f.} of the eigenvalues of complex Gaussian quadratic forms
and Wishart matrices, with arbitrary multiplicities for the
eigenvalues of the associated covariance matrix.
\item
Derivation of the ergodic capacity of single-user
MIMO systems that accounts for arbitrary power levels and arbitrary correlation across the
transmitting antenna elements, \mynote{ or arbitrary correlation at the receiver side}.
\item
Derivation of capacity expressions for \ac{MIMO} systems in the
presence of multiple MIMO interferers, valid for any number of interferers, each with
arbitrary number of antennas having possibly unequal power levels.

\end{itemize}

%

The paper is organized as follows: in Section~\ref{sec:sys} we
introduce the system model for multiuser MIMO setting, relating
the ergodic capacity of MIMO systems in the presence of multiple \ac{MIMO}
interferers to that of single-user MIMO systems with no interference.
General
results on hypergeometric functions of matrix arguments are
given in Section~\ref{sec:appendixhypfun}. The 
joint \ac{p.d.f.} of eigenvalues for Gaussian quadratic forms and Wishart matrices with arbitrary correlation is given   in Section~\ref{sec:quadforms}.
 In Section~\ref{sec:unifMIMOcap} we give a unified expression for the
capacity of single-user MIMO systems that accounts for arbitrary correlation matrix at one side. 
Numerical results for MIMO relay networks and multiuser MIMO are presented in Section~\ref{sec:MIMOCCIcap},  and conclusions are given in Section~\ref{sec:conclusion}.


 Throughout the
paper vectors and matrices are indicated by bold, $|\bf A|$ and
$\det \bf A$ denote the determinant of matrix $\bf A$, and
$a_{i,j}$ is the $\ith{(i,j)}$ element of $\bf A$.
Expectation operator is denoted by $\EX{\cdot}$, and in particular
$\EXs{X}{\cdot}$ denotes expectation with respect to the random
variable $X$. The superscript $\dag$ denotes conjugation and
transposition, $\I$ is the identity matrix (in particular $\I_n$
refers to the $(n \times n)$ identity matrix), $\tr{\bf A}$ is the trace of ${\bf A}$ and $\oplus$ is
used for the direct sum of matrices \mynote{defined as ${\bf A}\oplus {\bf B}=\diag{\bf A,B}$} \cite{HorJoh:B90}.

\section{System Models}\label{sec:sys}
 
We consider a network scenario as shown in
Fig.~\ref{fig:scenario}, where a MIMO-$(\Nt_0,\Nr)$ link\mynote{, with $\Nt_0$ and $\Nr$ denoting the numbers of transmitting and receiving antennas, respectively,}  is
subject to $\Ni$ MIMO co-channel interferers from other links,
each with arbitrary number of antennas.
\begin{figure}
\centerline{
\unitlength 0.90mm 
\linethickness{0.4pt}
\ifx\plotpoint\undefined\newsavebox{\plotpoint}\fi 
\begin{picture}(89,107)(0,0)
\put(10.2,103.4){\line(1,0){8}}
\put(22.2,106.4){\line(0,-1){6}}
\put(15.2,82.4){\makebox(0,0)[cc]{$\Nt_0$}}
\put(18.2,103.4){\line(4,3){4}}
\put(18.2,103.4){\line(4,-3){4}}
\put(10.45,88.15){\line(1,0){8}}
\put(18.45,88.15){\line(4,3){4}}
\put(14.95,95.65){\oval(2.5,19.5)[]}
\multiput(18.55,87.99)(.0525676,-.0333784){74}{\line(1,0){.0525676}}
\put(22.44,85.52){\line(0,1){5.3}}
\put(78,87.12){\line(-1,0){8}}
\put(78,103.12){\line(-1,0){8}}
\put(66,84.12){\line(0,1){6}}
\put(66,100.12){\line(0,1){6}}
\put(74,82.12){\makebox(0,0)[]{$\Nr$}}
\put(70,87.12){\line(-4,-3){4}}
\put(70,103.12){\line(-4,-3){4}}
\put(70,87.12){\line(-4,3){4}}
\put(70,103.12){\line(-4,3){4}}
\put(73.25,94.87){\oval(2.5,19.5)[]}
\put(4,94){\makebox(0,0)[cc]{$\Pt_0, \ \x_0$}}
\put(38,96){\makebox(0,0)[cc]{$\H_0$}}
\put(37,73){\makebox(0,0)[cc]{$\H_1$}}
\put(38,34){\makebox(0,0)[cc]{$\H_{\Ni}$}}
\put(55,96){\vector(1,0){.07}}
\put(44,96){\line(1,0){11}}
\put(32,70){\vector(3,2){.07}}
\multiput(27,67)(.0561798,.0337079){89}{\line(1,0){.0561798}}
\put(51,81){\vector(3,2){.07}}
\multiput(43,76)(.05369128,.03355705){149}{\line(1,0){.05369128}}
\put(32,26){\vector(3,4){.07}}
\multiput(26,19)(.03370787,.03932584){178}{\line(0,1){.03932584}}
\put(55,57){\vector(2,3){.07}}
\multiput(44,40)(.033639144,.051987768){327}{\line(0,1){.051987768}}
\put(89,94){\makebox(0,0)[cc]{$\y$}}
\put(10.2,25){\line(1,0){8}}
\put(10.2,73){\line(1,0){8}}
\put(22.2,28){\line(0,-1){6}}
\put(22.2,76){\line(0,-1){6}}
\put(15.2,4){\makebox(0,0)[cc]{$\Nt_{\Ni}$}}
\put(15.2,52){\makebox(0,0)[cc]{$\Nt_1$}}
\put(18.2,25){\line(4,3){4}}
\put(18.2,73){\line(4,3){4}}
\put(18.2,25){\line(4,-3){4}}
\put(18.2,73){\line(4,-3){4}}
\put(10.45,9.75){\line(1,0){8}}
\put(10.45,57.75){\line(1,0){8}}
\put(18.45,9.75){\line(4,3){4}}
\put(18.45,57.75){\line(4,3){4}}
\put(14.95,17.25){\oval(2.5,19.5)[]}
\put(14.95,65.25){\oval(2.5,19.5)[]}
\multiput(18.55,9.59)(.0525676,-.0333784){74}{\line(1,0){.0525676}}
\multiput(18.55,57.59)(.0525676,-.0333784){74}{\line(1,0){.0525676}}
\put(22.44,7.12){\line(0,1){5.3}}
\put(22.44,55.12){\line(0,1){5.3}}
\put(4,15.6){\makebox(0,0)[cc]{$\Pt_{\Ni}, \ \x_{\Ni}$}}
\put(4,63.6){\makebox(0,0)[cc]{$\Pt_1, \ \x_1$}}
\put(15,45){\circle*{0}}
\put(15,41){\circle*{2}}
\put(15,36){\circle*{2}}
\put(15,45){\circle*{2}}
\put(80,103){\circle{2.83}}
\put(80,87){\circle{2.83}}
\put(80,104){\line(0,-1){2}}
\put(80,88){\line(0,-1){2}}
\put(80,102){\line(0,1){0}}
\put(80,86){\line(0,1){0}}
\put(80,102){\line(0,1){0}}
\put(80,86){\line(0,1){0}}
\put(79,103){\line(1,0){2}}
\put(79,87){\line(1,0){2}}
\put(81,103){\line(1,0){3}}
\put(81,87){\line(1,0){3}}
\put(79,107){\line(0,1){0}}
\put(80,101){\vector(0,1){.07}}
\put(80,98){\line(0,1){3}}
\put(80,85){\vector(0,1){.07}}
\put(80,82){\line(0,1){3}}
\put(80,78){\makebox(0,0)[cc]{$\n$}}
\put(33,96){\vector(1,0){.07}}
\put(27,96){\line(1,0){6}}
\end{picture}
} \caption{MIMO Network.} \label{fig:scenario}
\end{figure}
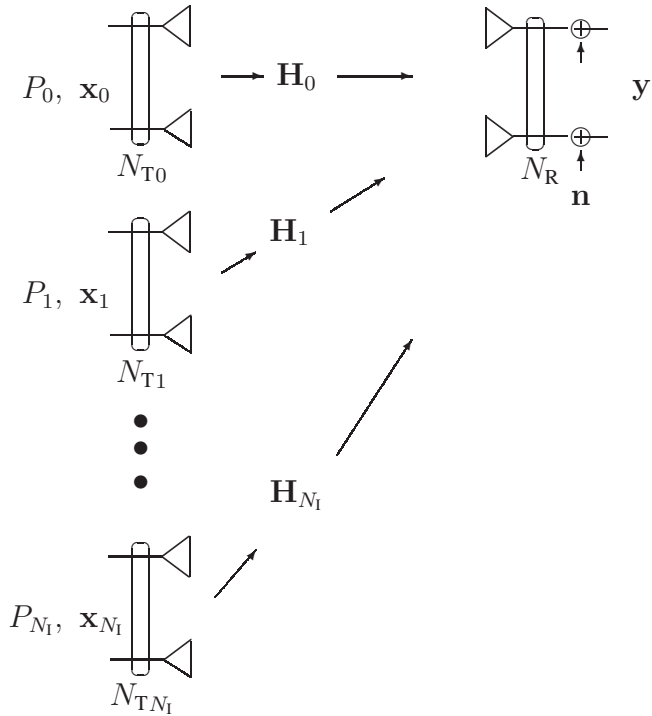
The $\Nr$-dimensional equivalent lowpass signal $\bf y$, after
matched filtering and sampling, at the output of the receiving
antennas can be written as 
\begin{equation}
\label{eq:y}
 \y= \H_0 \, \x_0 + \sum_{k=1}^{\Ni} \H_k \, \x_k+ \n
\end{equation}
where $\x_0, \x_1, \ldots, \x_{\Ni} $ denote the complex
transmitted vectors with dimensions $\Nt_0, \Nt_1, \ldots,
\Nt_{\Ni}$, respectively. Subscript $0$ is used for the desired
signal, while subscripts $1,\ldots ,\Ni$ are for the interferers.
The additive noise 
$\n$ is an $\Nr$-dimensional random vector with zero-mean
\ac{i.i.d.}  circularly symmetric complex Gaussian entries, each
with independent real and imaginary parts having variance $\sigma^2/2$,
so that $\EX{\n \n^{\dag}}=\sigma^2 \I$. The power transmitted
from the $\ith{k}$ user is $\EX{\x_k^{\dag}\x_k}=\Pt_k$.

The matrices $\H_k$ in \eqref{eq:y} denote the channel matrices of size $(\Nr \times \Nt_k)$ 
with complex elements $h^{(k)}_{i,j}$ describing the gain of the
radio channel between the $\ith{j}$ transmitting antenna of the
$\ith{k}$ MIMO interferers and the $\ith{i}$ receiving antenna of
the desired link. In particular, $\H_0$ is the matrix describing
the channel of the desired link (see Fig.~\ref{fig:scenario}).

When considering statistical variations of the channel, the
channel gains must be described as \acp{r.v.}.
In particular, we assume uncorrelated MIMO Rayleigh fading channels for which the
entries of $\H_k$ are \ac{i.i.d.} circularly symmetric complex
Gaussian \acp{r.v.} with zero-mean
and variance one, i.e., $\EX{|h^{(k)}_{i,j}|^2}=1$. 
%
\mynote{
With this normalization, $\Pt_k$ represents the short-term average received power per antenna element 
from user $k$, which depends on the transmit power, path-loss, and shadowing between transmitter $k$ and the (interfered) receiver.  Thus, the  $\Pt_k$ are in general different.
}


Conditioned to the channel matrices
$\left\{\H_k\right\}_{k=0}^{\Ni}$, the mutual information between
the received vector, $\y$, and the desired transmitted vector,
$\x_0$, is \cite{CovTho:B91} 
%
%

\begin{equation}
\label{eq:Ixy1} {\mathcal{I}}
\left(\x_0 \ ;\ \y \ | \left\{\H_k\right\}_{k=0}^{\Ni} \right)= \mathcal{H}\left(\y \ | \left\{\H_k\right\}_{k=0}^{\Ni} \right)-\mathcal{H}\left({\y \ | \ \x_0}, \left\{\H_k\right\}_{k=0}^{\Ni} \right)
\end{equation}
where ${\mathcal{H}({\cdot})}$ denotes differential entropy.

Here we consider the scenario in which the receiver has perfect
\ac{CSI}, and all the transmitters have no \ac{CSI}. Note
that
 the term \ac{CSI} includes the information about the
channels associated with all other MIMO interfering users. 
In this case, 
\mynote{since the users do not know what is the interference seen at the receiver (if any), a reasonable strategy is that 
each user transmits 
} 
 circularly
symmetric Gaussian vector signals with zero mean and \ac{i.i.d.}
elements. Thus, the transmit
power per antenna element of the $\ith{k}$ user is $\Pt_k/\Nt_k$. 
%
Note that this model includes the case in which the power levels 
 of the individual antennas are different: it suffices to decompose a
transmitter into virtual sub-transmitters, each with the proper
power level.

Hence, conditioned on all channel matrices
$\left\{\H_k\right\}_{k=0}^{\Ni}$ in \eqref{eq:y}, both $\y$ and
$\y| \x_0 $ are circularly symmetric Gaussian. Since the differential entropy of a
Gaussian vector is proportional to the logarithm of the
determinant of its covariance
matrix, 
we obtain the conditional \mynote{mutual information}
%
\begin{equation}
\label{eq:Ixy2}
 C_{\mathrm{MU}}\left(\left\{\H_k\right\}_{k=0}^{\Ni}\right)=\log\frac{\det \mathbf{K}_{\y}}{\det \mathbf{K}_{\y |
 \x_0}}
\end{equation}
where $\mathbf{K}_{\y}$ 
 and $\mathbf{K}_{\y |\x_0}
$ respectively denote the covariance matrices of $\y$ and $\y|
\x_0 $, conditioned on the channel gains
$\left\{\H_k\right\}_{k=0}^{\Ni}$.
By expanding the covariance matrices using \eqref{eq:y}, the
conditional  \mynote{mutual information} of a MIMO link in the presence of multiple
MIMO interferers with \ac{CSI} only at the receiver is then given
by:
\begin{equation}
\label{eq:CapCond}
C_{\mathrm{MU}}\left(\left\{\H_k\right\}_{k=0}^{\Ni}\right)=
\log \frac{\det\left({\bf I}_{\Nr} +{\bf \tilde{H} \tilde{\Psi}
\tilde{H}^{\dag} } \right)}{\det\left({\bf I}_{\Nr}+ {\bf H \Psi
H^{\dag}}\right)}
\end{equation}
where the $\Nr \times (\sum_{i=1}^{\Ni}\Nt_i)$ matrix $\H$ is
$$\H=\left[\H_1|\H_2| \cdots |\H_{\Ni}\right]$$
the $\Nr \times (\sum_{i=0}^{\Ni}\Nt_i)$ matrix $\tilde{\H}$ is
$$\tilde{\H}=\left[\H_0|\H\right]$$
the covariance matrices ${\bf {\Psi}}, {\bf \tilde{\Psi}}$ are


\begin{equation}
\label{eq:MatrixPsi}
{\bf \Psi}=\r_1 \, \I_{\Nt_1} \oplus \r_2 \, \I_{\Nt_2} \oplus \cdots \oplus \r_{\Ni} \, \I_{\Nt_{\Ni}} 
\end{equation}
and
\begin{equation}
\label{eq:MatrixPsiTilde}
{\bf \tilde{\Psi}}= \r_0 \, \I_{\Nt_0} \oplus {\bf \Psi} 
\end{equation}
with
\begin{equation}\label{eq:ri}
\r_i= \frac{\Pt_i}{\Nt_i \sigma^2}.
\end{equation}


\mynote{With random channel matrices the 
 mutual information in
\eqref{eq:CapCond}  is 
the difference between} random variables of the form $\log\det\left({\bf
I}+ {\H {\bf \Phi} \H^{\dag} }\right)$ where the elements of $\H$
are \ac{i.i.d.} complex Gaussian and ${\bf \Phi}$ is a covariance
matrix. The statistics of such random variables has been
investigated in
\cite{ChiWinZan:J03,SmiRoySha:J03,ShiWinLeeChi:J05}, assuming that
the eigenvalues of ${\bf \Phi}$ were distinct. However, in the
scenario under analysis these results cannot be used directly,
since in \eqref{eq:CapCond} each eigenvalue $\r_i$ of ${\bf \Psi}$ and $\tilde{\mathbf{\Psi}}$ has multiplicity $\Nt_i$.

We consider  
the ergodic  \mynote{mutual information} as a performance measure: 
taking the expectation of \eqref{eq:CapCond} with respect to the
distribution of $\left\{\H_k\right\}_{k=0}^{\Ni}$, we get
\begin{equation}\label{eq:meanCapMIMOCCI}
\mathcal{C}_{\mathrm{MU}}\triangleq\EX{C_{\mathrm{MU}}\left(\left\{\H_k\right\}_{k=0}^{\Ni}\right)}=\mathcal{C}_{\mathrm{SU}}\left(\sum_{i=0}^{\Ni}\Nt_i,
\Nr, {\bf
\tilde{\Psi}}\right)-\mathcal{C}_{\mathrm{SU}}\left(\sum_{i=1}^{\Ni}\Nt_i,
\Nr, {\bf \Psi}\right)
\end{equation}
where $\mathcal{C}_{\mathrm{SU}}\left(\nt, \nr, {\bf \Phi}\right) \triangleq \EXs{\H}{\log\det\left({\bf I}_{\nr}+ {\H {\bf \Phi}
\H^{\dag} }\right)}$
denotes the ergodic \mynote{mutual information} 
 of a single-user MIMO-$(\nt, \nr)$
Rayleigh fading channel with unit noise variance per receiving antenna and channel covariance matrix ${\bf \Phi}$ at the transmitter.

Note that the ``building block''
$
\EXs{\H}{\log\det\left({\bf I}+ {\H {\bf \Phi}
\H^{\dag} }\right)}$ is simple to evaluate when the covariance
matrix ${\bf \Phi}$ is proportional to an identity matrix, which
corresponds to a typical interference-free case with equal
transmit power among all transmitting antennas (see, e.g.,
\cite{Tel:99}). In contrast, in the presence of interference, the
covariance matrix is of the type indicated in \eqref{eq:MatrixPsi}
and  \eqref{eq:MatrixPsiTilde}, where the power
levels of the different users are in general different. 
Note that even when the power for the $\ith{i}$ user is equally
spread over the $\Nt_i$ antennas, the matrices 
 in \eqref{eq:MatrixPsi} and \eqref{eq:MatrixPsiTilde} 
  are in general not proportional to identity matrices and their
eigenvalues have multiplicities greater than one.
%
Therefore, studying MIMO systems in the presence of multiple MIMO
cochannel interferers requires the characterization of
$\mathcal{C}_{\mathrm{SU}}\left(\nt, \nr, {\bf \Phi}\right)$ in a general setting in which the covariance matrix ${\bf \Phi}$ has eigenvalues of arbitrary multiplicities.

To this aim, we derive \mynote{in the next sections} simple expressions for the hypergeometric
functions of matrix arguments with not necessarily distinct
eigenvalues
; then, we obtain the joint \ac{p.d.f.} of the eigenvalues of central Wishart matrices as well as that of Gaussian quadratic forms with arbitrary covariance matrix.
%

\mynote{\section{Hypergeometric functions with matrix arguments having arbitrary eigenvalues} \label{sec:appendixhypfun}
}


Hypergeometric functions with matrix arguments 
\cite{Jam:64} have been used extensively in multivariate statistical analysis, especially in problems related to the distribution of random matrices \cite{Mui:B82}. These functions are defined in terms of a series of zonal polynomials, and, as such, they are functions only of the eigenvalues (or latent roots) of the argument matrices \cite{Jam:64,Mui:B82}.

\begin{definition} The hypergeometric functions of two Hermitian $m \times m$ matrices ${\bf \Lambda}$ and ${\bf W}$ are defined by \cite{Jam:64}
%

%
\begin{equation}\label{eq:pFqXY}
{}_p\td{F}_q\left(a_1,\ldots,a_p; b_1,\ldots,b_q;{\bf \Lambda}, {\bf W}\right)\triangleq
\sum_{k=0}^{\infty}\sum_{\kappa}\frac{(a_1)_{\kappa}\cdots
(a_p)_{\kappa}}{(b_1)_{\kappa}\cdots
(b_q)_{\kappa}}\frac{C_{\kappa}({\bf \Lambda}) C_{\kappa}({\bf
W})}{k! {C_\kappa}({\bf I}_m) }
\end{equation}
where $C_{\kappa}(\cdot)$ is a symmetric homogeneous
polynomial of degree $k$ in the eigenvalues of its argument,
called {\it zonal polynomial}, the sum $\sum_{\kappa}$ is over all
partitions of $k$, i.e., $\kappa=(k_1,\ldots,k_m)$ with $k_1\ge
k_2 \ge \cdots \ge k_m \ge 0$, $k_1+k_2+\cdots+k_m=k$, and the
generalized hypergeometric coefficient $(a)_{\kappa}$ is given by
$(a)_{\kappa}=\prod_{i=1}^{m}\left(a-\frac{1}{2}(i-1)\right)_{k_i}$
 with $(a)_k=a (a+1) \cdots (a+k-1)$, $(a)_0=1$.
\end{definition}

We remark that zonal polynomials are symmetric polynomials in the
eigenvalues of the matrix argument. Therefore,  hypergeometric
functions are only functions of the
 eigenvalues of their matrix arguments. In other words,
without loss of generality we can replace ${\bf \Lambda}$ and ${\bf W}$ 
with the diagonal matrices $\diag{\lambda_1,\ldots\lambda_m}$ and
$\diag{w_1,\ldots w_m}$, where $\lambda_i$ and $w_j$ are the
eigenvalues of ${\bf \Lambda}$ and ${\bf W}$, respectively.
Clearly the order of ${\bf \Lambda}$ and ${\bf W}$ is unimportant.


It is quite evident that these functions expressed as a series of
zonal polynomials are in general very difficult to manage and the
form of \eqref{eq:pFqXY} is not tractable for further analysis.
Fortunately, when the eigenvalues of ${\bf \Lambda}$ and ${\bf W}$
are all distinct, a simpler expression  
in terms of determinants of matrices whose elements are
hypergeometric functions of scalar arguments can be obtained as follows \cite[Lemma 3]{Kha:70}:

\begin{lemma}\label{lemma: lemma3khatri} ([Khatri, 1970])
Let ${\bf\Lambda}=\diag{\lambda_1,\ldots\lambda_m}$ and ${\bf
W}=\diag{w_1,\ldots w_m}$ with $\lambda_1>\cdots>\lambda_m$ and
$w_1>\cdots>w_m$. Then we have
%
%
%
%
\begin{equation}\label{eq:pFqXYkhatri}
{}_p\td{F}_q\left(a_1,\ldots,a_p; b_1,\ldots,b_q; {\bf \Lambda}, {\bf W}\right)=
\Gamma_{(m)}(m)\frac{\psi_q^{(m)}(b)}{\psi_p^{(m)}(a)}
\frac{\left|{\bf G}\right|}{\prod_{i<j}\left(\lambda_i-\lambda_j\right)\prod_{i<j}\left(w_i-w_j\right)}
\end{equation}
where $\Gamma_{(m)}(n) \triangleq \prod_{i=1}^{m} (n-i)! $, 
$\psi_q^{(m)}(b)=\prod_{i=1}^{m}\prod_{j=1}^{q}(b_j-i+1)^{i-1}$
and the $\ith{ij}$ element of the ($m\times m$) matrix ${\bf
G}
$ is defined in terms of
hypergeometric functions of scalar arguments as follows
\begin{equation}\label{eq:Gkhatri}
g_{i,j}={}_p{F}_q\left(\tilde{a}_1,\ldots,\tilde{a}_p;
\tilde{b}_1,\ldots,\tilde{b}_q; \lambda_i w_j \right)
\end{equation}
with $\tilde{a}_i=a_i-m+1$ and $\tilde{b}_i=b_i-m+1$.

\noindent Important particular cases are
\begin{equation}\label{eq:0F0XYkhatri}
{}_0\td{F}_0\left({\bf \Lambda}, {\bf W}\right)=
\Gamma_{(m)}(m)\frac{\left|{\bf G}_0\right|}{\prod_{i<j}\left(\lambda_i-\lambda_j\right)\prod_{i<j}\left(w_i-w_j\right)}
\end{equation}
and
\begin{equation}\label{eq:1F0XYkhatri}
{}_1\td{F}_0\left(r;{\bf \Lambda}, {\bf W}\right)=
\frac{\Gamma_{(m)}(m)}{\psi_1^{(m)}(r)}\frac{\left|{\bf
G}_1
\right|}{\prod_{i<j}\left(\lambda_i-\lambda_j\right)\prod_{i<j}\left(w_i-w_j\right)}
\end{equation}
where the $\ith{ij}$ elements of ${\bf G}_0$ and ${\bf G}_1$ are
given by $e^{\lambda_i w_j}$ and $\left(1-\lambda_i w_j\right)^{m-r-1}$, respectively.
\end{lemma}
\medskip

%

These expressions 
have been recently used 
 to study the distribution of Gaussian quadratic forms, 
to express the \ac{p.d.f.} of the eigenvalues of Wishart matrices, and to analyze the 
information-theoretic capacity and error rates of communication systems involving multiple antennas 
\cite{GaoSmi:00,Chi:02,ChiWinZan:J03,ChiWinZanMalWint:J03,ChiWinZan:J03a,ZanChiWin:J05,McKCol:05,SmiRoySha:J03,ShiWinLeeChi:J05,ChiWinZan:J05,ZanChiWin:C05}.
However, it is important to underline that Lemma~\ref{lemma: lemma3khatri}
requires the eigenvalues of the matrices to be all distinct.

\medskip

Here, we generalize Lemma~\ref{lemma: lemma3khatri} to include the case where the eigenvalues are not necessarily distinct. To this aim we first need the following lemma.

\begin{lemma}\label{lemma: lemma0chiani}

Let $P: {\bf A} \rightarrow \mathbb{R}$ be defined over ${\bf A}
\subset \mathbb{R}^m$ as follows:

\begin{equation}\label{eq:defPw1wm}
P(w_1,\ldots,w_m)\triangleq \frac{1}{\prod_{i<j}(w_i-w_j)}
  \cdot \left|\begin{array}{llll} 
    f_1(w_1) & f_1(w_2) & \cdots & f_1(w_m) \\
                      &                   &        &         \\
    \vdots            & \vdots            & \cdots & \vdots  \\
                      &                   &        &         \\
    f_m(w_1) & f_m(w_2) & \cdots & f_m(w_m) \\
  \end{array}\right|\,
\end{equation}
where $w_1> w_2 \cdots > w_m$, and the functions $f_i(w)$ have
derivatives $f^{(n)}_i(w)=\frac{d^n f_i(w)}{d w ^n}$ of orders at
least $m-1$ throughout neighborhoods of the points $w_1, \ldots,
w_m$.

Then, the continuous extension $\breve{P}(w_1,w_2,\ldots,w_m)$ of
the function $P(w_1,w_2,\ldots,w_m)$ to those points in
$\mathbb{R}^m$ with $L$ coincident arguments $w_K=w_{K+1}=\cdots
w_{K+L-1}$ is obtained by removing the zero factors  from the
denominator in \eqref{eq:defPw1wm}, replacing the columns of the
matrix in \eqref{eq:defPw1wm} corresponding to the coincident
arguments with the successive derivatives $f^{(L-l)}_i(w_K),$ $l=1,\ldots,L$, and dividing by a scaling
factor $\Gamma_{(L)}(L)=\prod_{i=1}^{L-1} i!$.

For example, for $w_1=w_2=\cdots w_L$, this procedure gives

\bigskip

\lefteqn{
\breve{P}(w_1,w_2,\ldots,w_m)=\frac{1}{\prod_{i<j, w_i \neq
w_j}(w_i-w_j) \prod_{i=1}^{L-1} i!}\cdot}

\begin{equation}\label{eq:theorem0}
\cdot  \left|\begin{array}{llllllll} 
    f^{(L-1)}_1(w_1) & f^{(L-2)}_1(w_1) & \cdots & f_1(w_1) & f_1(w_{L+1}) & \cdots & f_1(w_m) \\
                      &       & & &             &    &             \\
    \vdots            & \vdots &&             & \vdots & \cdots & \vdots  \\
                      &       & &              &      &  &         \\
    f^{(L-1)}_m(w_1) & f^{(L-2)}_m(w_1) & \cdots & f_m(w_1) & f_m(w_{L+1}) & \cdots & f_m(w_m) \\
  \end{array}\right|\,.
\end{equation}
%
%
More generally,
a similar expression is valid if there are more groups of coinciding arguments: in this case,
 for each group of coincident arguments $w_K=\ldots =w_{K+L-1}$ the correspondent columns of the matrix in
\eqref{eq:defPw1wm} are to be replaced by  $f^{(L-l)}_i(w_K),$ $l=1,\ldots,L$, with a scaling factor
$\prod_{i=1}^{L-1} i!$.

\end{lemma}

\bigskip

\begin{proof}  \mynote{See Appendix \ref{app:proofs}.} 
\end{proof}

\medskip

With Lemma \ref{lemma: lemma0chiani} we can now generalize
\eqref{eq:0F0XYkhatri}, \eqref{eq:1F0XYkhatri} and, more
generally, \eqref{eq:pFqXYkhatri}.

\begin{lemma}\label{lemma: lemma1chiani}

Let ${\bf\Lambda}=\diag{\lambda_1,\ldots\lambda_m}$ and ${\bf
W}=\diag{w_1,\ldots w_m}$ with $\lambda_1>\cdots>\lambda_m$ and
$w_1>\cdots>w_k=w_{k+1}=\cdots=w_{k+L-1}>w_{k+L}>\cdots>w_m$. 
Then we have\footnote{From here on we will use the same symbols for the
functions \eqref{eq:0F0XYkhatri}, \eqref{eq:1F0XYkhatri},
\eqref{eq:pFqXYkhatri} and their continuous extension.}
\begin{equation}\label{eq:0F0XYchiani}
{}_0\td{F}_0\left({\bf \Lambda}, {\bf W}\right)=
\frac{\Gamma_{(m)}(m)}{\Gamma_{(L)}(L)}\frac{\left|{\bf G
}\right|}{\prod_{i<j}\left(\lambda_i-\lambda_j\right)\prod_{i<j ,
\, w_i\ne w_j }\left(w_i-w_j\right)}
\end{equation}
where the elements of ${\bf G}$ are
\begin{equation}\label{eq:G0F0chiani}
g_{i,j}=\left\{
  \begin{array}{ll}
    \lambda_i^{L-1+k-j} e^{\lambda_i w_k} & \qquad j=k,\ldots,k+L-1 \\
    e^{\lambda_i w_j} & \qquad \text{elsewhere}
  \end{array}
  \right. 
\end{equation}
that is, the matrix $\bf G$ is the same as that appearing in
\eqref{eq:0F0XYkhatri} except that the $L$ columns corresponding to the coincident eigenvalues are $\lambda^{L-1}_i e^{\lambda_i w_k}, \lambda^{L-2}_i e^{\lambda_i w_k}, \ldots, \lambda^{2}_i
e^{\lambda_i w_k}, \lambda_i e^{\lambda_i w_k}, e^{\lambda_i
w_k}$.

%
\end{lemma}
\begin{proof}
The proof is immediate by direct application of Lemma~\ref{lemma:
lemma0chiani} with $f_i(w)=e^{\lambda_i w}$.

\end{proof}

Lemma \ref{lemma: lemma1chiani} can be directly extended to more groups of
coincident eigenvalues. In general, the rule is that each eigenvalue
$w$ of multiplicity $L>1$ gives rise to $L$ columns $\lambda^{L-1}_i
e^{\lambda_i w},$ $\lambda^{L-2}_i e^{\lambda_i w}, \ldots,$ $
\lambda^{2}_i e^{\lambda_i w},\lambda_i e^{\lambda_i w},
e^{\lambda_i w}$  in the matrix ${\bf G}$ of \eqref{eq:0F0XYchiani},
with the proper scaling factor $\Gamma_{(L)}(L)$.

\medskip

Using Lemma~\ref{lemma: lemma1chiani} with $k=m-L+1$ and $w_k=0$ results in the following corollary, for the  case where some eigenvalues are equal to zero.

\begin{corollary} \label{corollary:lemma1chiani}
Let ${\bf\Lambda}=\diag{\lambda_1,\ldots\lambda_m}$ and ${\bf
W}=\diag{w_1,\ldots w_m}$ with $\lambda_1>\cdots>\lambda_m$ and
$w_1>\cdots>w_{m-L+1}=w_{m-L+2}=\cdots=w_{m}=0$. 
Then we have
{
\begin{equation}\label{eq:0F0XYchianiZeroEigen}
{}_0\td{F}_0\left({\bf \Lambda}, {\bf W}\right)=
\frac{\Gamma_{(m)}(m)}{\Gamma_{(L)}(L)}\frac{\left|{\bf G
}\right|}{\prod_{i<j}\left(\lambda_i-\lambda_j\right)\prod_{i<j\le
m-L}\left(w_i-w_j\right)\prod_{i=1}^{m-L}w^L_i} 
\end{equation}
}
where the elements of ${\bf G}$ are as follows
\begin{equation}\label{eq:G0F0chianiZeroEigen} g_{i,j}=\left\{
  \begin{array}{ll}
    \lambda_i^{m-j} & \qquad j=m-L+1,\ldots,m \\
    e^{\lambda_i w_j} & \qquad \text{elsewhere} .
  \end{array}
  \right. 
\end{equation}
%
\end{corollary}

\medskip
%

We can apply a similar methodology to derive the general expression
for ${}_1\td{F}_0(\cdot; \cdot,\cdot)$, as in the following Lemma.

\medskip

\begin{lemma}\label{lemma: lemma2chiani}
Let ${\bf\Lambda}=\diag{\lambda_1,\ldots\lambda_m}$ and ${\bf
W}=\diag{w_1,\ldots w_m}$ with $\lambda_1>\cdots>\lambda_m$ and
$w_1>\cdots>w_k=w_{k+1}=\cdots=w_{k+L-1}>w_{k+L}>\cdots>w_m$. 
 Then we have
\begin{equation}\label{eq:1F0XYchiani}
{}_1\td{F}_0\left(r; {\bf \Lambda}, {\bf
W}\right)=\frac{\Gamma_{(m)}(m)}{\Gamma_{(L)}
(L)}\frac{(-1)^{(L-1)L/2}}{\psi_1^{(m)}(r)} \cdot \frac{
\gamma^{L-1}(\gamma-1)^{L-2}\cdots(\gamma-L+2)\,\, \left|{\bf A
}\right| }{\prod_{i<j}\left(\lambda_i-\lambda_j\right)\prod_{i<j ,
\, w_i\ne w_j }\left(w_i-w_j\right)} 
\end{equation}
where $\gamma=m-r-1$ and the ($m\times m$) matrix ${\bf A}$ has
elements as follows
{
\begin{equation}\label{eq:A1F0chiani}
a_{i,j}=\left\{
  \begin{array}{ll}
    \lambda_i^{L-1+k-j} \left(1-\lambda_i w_j\right)^{\gamma-(L-1+k-j)}  & j=k,\ldots,k+L-1 \\
    \left(1-\lambda_i w_j\right)^{\gamma} & \text{elsewhere.}
  \end{array}
  \right.
\end{equation}
}
In other words, the matrix $\bf A$ is the same as that appearing in
\eqref{eq:1F0XYkhatri}, except that the $L$ columns corresponding to the $L$ coincident eigenvalues are
$\lambda^{L-1}_i \left(1-\lambda_i w_k\right)^{\gamma-(L-1)},$ $
\ldots, \lambda_i \left(1-\lambda_i w_k\right)^{\gamma-1},$ $
\left(1-\lambda_i w_k\right)^{\gamma}$.
%
%
\end{lemma}

\medskip

\begin{proof}
For the proof we apply Lemma~\ref{lemma: lemma0chiani} with
$f_i(w)=\left(1-\lambda_i w_k\right)^{\gamma}$, whose $\ith{n}$
derivative is $f^{(n)}_i(w)=(-\lambda_i)^n \gamma (\gamma-1)\cdots
(\gamma-n+1) \left(1-\lambda_i w_k\right)^{\gamma-n}$.
\end{proof}

\medskip
Lemma~\ref{lemma: lemma2chiani} can be further generalized to more
groups of coincident eigenvalues: each eigenvalue $w$ of
multiplicity $L>1$ gives rise to $L$ columns $\lambda^{L-1}_i
(1-\lambda_i w)^{\gamma-(L-1)}, \ldots, \lambda^{2}_i (1-\lambda_i
w)^{\gamma-2},\lambda_i (1-\lambda_i w)^{\gamma-1}, (1-\lambda_i
w)^{\gamma}$ in the matrix ${\bf A}$ of \eqref{eq:1F0XYchiani}, and to
a factor $(-1)^{(L-1)L/2} \gamma^{L-1}\cdots
(\gamma-L+2)/\Gamma_{(L)}(L)$.

\medskip

Using Lemma~\ref{lemma: lemma2chiani} with $k=m-L+1$ and $w_k=0$ results in the following corollary. 

\begin{corollary} \label{corollary:lemma2chiani}
Let ${\bf\Lambda}=\diag{\lambda_1,\ldots\lambda_m}$ and ${\bf
W}=\diag{w_1,\ldots w_m}$ with $\lambda_1>\cdots>\lambda_m$ and
$w_1>\cdots>w_{m-L+1}=w_{m-L+2}=\cdots=w_{m}=0$. 
Then we have that \eqref{eq:1F0XYchiani} holds, with
\begin{equation}\label{eq:A1F0chianiZeroEigen} a_{i,j}=\left\{
  \begin{array}{ll}
    \lambda_i^{m-j} & j=m-L+1,\ldots,m \\
    (1-\lambda_i w_j)^{\gamma} & \text{elsewhere.}
  \end{array}
  \right.
\end{equation}
In other words, the matrix ${\bf A}$ has in this case the last $L$
columns with elements $\lambda_i^{L-1},
\lambda_i^{L-2},\ldots,\lambda_i,1$.

\end{corollary}

\medskip


Finally, we give the result for the ${}_p\td{F}_q(\cdot)$.

\begin{lemma}\label{lemma: lemma3chiani}
Let ${\bf\Lambda}=\diag{\lambda_1,\ldots\lambda_m}$ and ${\bf
W}=\diag{w_1,\ldots w_m}$ with $\lambda_1>\cdots>\lambda_m$ and
$w_1>\cdots>w_k=w_{k+1}=\cdots=w_{k+L-1}>w_{k+L}>\cdots>w_m$. 
 Then we have
\begin{equation}\label{eq:pFqXYchiani}
{}_p\td{F}_q\left(a_1,\ldots,a_p; b_1,\ldots,b_q; {\bf \Lambda},
{\bf W}\right)= \Xi \, \frac{\left|{\bf
C}\right|}{\prod_{i<j}\left(\lambda_i-\lambda_j\right)\prod_{i<j,
w_i\neq w_j}\left(w_i-w_j\right)} 
\end{equation}
where the ($m\times m$) matrix ${\bf C}$ has elements as follows
%
{
\begin{equation}\label{eq:CpFqchiani}
c_{i,j}=
    \lambda_i^{L-1+k-j}{}_p\td{F}_q\left(a_1-m+L+k-j, \ldots, b_q-m+L+k-j;\lambda_i w_j\right)
\end{equation}
} for $ j=k,\ldots,k+L-1$, and
$$c_{i,j}={}_p\td{F}_q\left(\tilde{a}_1,\ldots,\tilde{a}_p; \tilde{b}_1,\ldots,\tilde{b}_q; \lambda_i
w_j \right)$$ elsewhere.
In \eqref{eq:pFqXYchiani} the constant $\Xi$ is

{
$$\Xi=\frac{\Gamma_{(m)}(m)}{\Gamma_{(L)}(L)}\frac{\psi_q^{(m)}(b)}{\psi_q^{(m)}(a)} \prod_{i=1}^{L-1} \frac{(\tilde{a}_1)_i (\tilde{a}_2)_i \cdots (\tilde{a}_p)_i}{(\tilde{b}_1)_i (\tilde{b}_2)_i \cdots
(\tilde{b}_q)_i}.
$$
}
\end{lemma}

\medskip

\begin{proof} \mynote{See Appendix \ref{app:proofs}.}
\end{proof}

\mynote{
\section{Gaussian quadratic forms with covariance matrix having eigenvalues of arbitrary multiplicity}
\label{sec:quadforms}


We now derive the joint \ac{p.d.f.} of the eigenvalues for Gaussian quadratic forms and central Wishart matrices with arbitrary one-sided correlation matrix.

\begin{lemma}\label{lemma: pdfquadraticforms}
Let $\mathbf{H}$ be a complex Gaussian $(p \times n)$ random matrix with zero-mean, unit variance, \ac{i.i.d.} entries and let $\mathbf{\Phi}$ be an $(n \times n)$
positive definite matrix. The joint \ac{p.d.f.} of the (real) non-zero ordered
eigenvalues $\lambda_1 \geq \lambda_2 \geq \ldots \geq
\lambda_{\nmin} \geq 0$ of the $(p \times p)$  quadratic form $\mathbf{W}=\mathbf{H}\mathbf{\Phi}\mathbf{H}^{\dag}$ is given by
\begin{equation} \label{eq:jpdfgeneralquadraticform}
f_{\bl{}} (x_{1}, \ldots, x_{\nmin})
    = K
    \left|{\bf V(x)}\right| \,
        \left|\tilde{\bf G}({\bf x,\bm \mu})\right| \prod_{i=1}^{\nmin} x_i^{p-\nmin}
\end{equation}
where $\nmin=\min(n,p)$, ${\bf V(x)}$ is the ($\nmin \times \nmin$) Vandermonde
matrix with elements $v_{i,j}=x^{i-1}_j$, 
\begin{equation}\label{eq:Kquadformgen} 
K= \frac{(-1)^{p(n-\nmin)}}{\Gamma_{(\nmin)}(p)} \, \frac{\prod_{i=1}^{L}\mu_{(i)}^{m_i p}}{ \prod_{i=1}^{L} \Gamma_{(m_i)}(m_i) \prod_{i<j}\left(\mu_{(i)}-\mu_{(j)}\right)^{m_i m_j}}
\end{equation}
and $\mu_{(1)} > \mu_{(2)} \ldots > \mu_{(L)}$ are the $L$
distinct eigenvalues of ${\bf \Phi}^{-1}$, with corresponding 
multiplicities $m_1, \ldots , m_L$ such that $\sum_{i=1}^{L}
m_i=n$. 

The ($n\times n$) matrix $\tilde{\bf G}({\bf x,\bm\mu})$ has elements 
\begin{equation} \label{eq:gtilde}
\tilde{g}_{i,j}=\left\{
  \begin{array}{ll}
    \left(-x_j \right)^{d_i} e^{-\mu_{(e_i)} x_j} & \qquad j=1,\ldots, \nmin \\
    \left[n-j\right]_{d_i} \mu_{(e_i)}^{n-j-d_i}& \qquad j=\nmin+1,\ldots,n \\
  \end{array}
  \right. 
\end{equation}
where $[a]_k = a(a-1)\cdots (a-k+1)$, $[a]_0 = 1$, $e_i$ denotes the unique integer such that
$$m_1+\ldots+m_{e_i-1}< i \leq m_1 +\ldots+m_{e_i} $$
and 
$$d_i=\sum_{k=1}^{e_i}m_k -i .$$

\end{lemma}
\begin{proof} \mynote{See Appendix \ref{app:proofs}.}
\end{proof}

Note that Lemma \ref{lemma: pdfquadraticforms} gives, in a compact form, the general joint distribution for the eigenvalues of a central Wishart ($p\geq n$), and central pseudo-Wishart or quadratic form ($n\geq p$), with arbitrary one-sided correlation matrix with not-necessarily distinct eigenvalues.

In fact, Lemma \ref{lemma: pdfquadraticforms} can be used for both $p\geq n$ and $n\geq p$; in particular, for $n\geq p$  we have $\prod_{i=1}^{\nmin} x_i^{p-\nmin}=1$ in \eqref{eq:jpdfgeneralquadraticform}, while for $p\geq n$ the second row in \eqref{eq:gtilde} disappears and $(-1)^{p(n-\nmin)}=1$ in \eqref{eq:Kquadformgen}.
}

Moreover, using Lemma \ref{lemma: pdfquadraticforms} and the results in \cite{ZanChiWin:J09,ChiZan:C08} we can also derive the marginal distribution of individual eigenvalues or of an arbitrary subset of the eigenvalues.




\section{\mynote{Ergodic Mutual information} 
 of a single-user MIMO system} 
 \label{sec:unifMIMOcap}
In this section we provide a unified analysis of the ergodic \mynote{mutual information} of a single-user
\ac{MIMO} system with arbitrary power levels/correlation among the
transmitting antenna elements or arbitrary correlation at the receiver, admitting correlation matrices with not-necessarily distinct eigenvalues. %

\mynote{Let us consider the function 
\begin{equation}\label{eq:defCgeneric}
\mathcal{C}_{\mathrm{SU}}\left(n, p, {\bf
\Phi}\right)=\EXs{\mathbf{H}}{\log\det\left({\bf I}_{p}+ {\H
{\bf \Phi} \H^{\dag} }\right)}
\end{equation}
where ${\bf \Phi}$ is a generic $(n \times n)$ positive definite matrix 
and $\mathbf{H}$ is a $(p \times n)$ random matrix with  zero-mean, unit variance complex Gaussian \ac{i.i.d.} entries.

Now, consider a single-user MIMO-$(\nt, \nr)$ Rayleigh fading channel 
 with ${\bf \Psi}_{\text{T}}, {\bf \Psi}_{\text{R}}$ denoting the $(\nt \times \nt)$ transmit and $(\nr \times \nr)$ receive correlation matrices, respectively, having diagonal elements equal to one. Assume the transmit vector ${\bf x}$ is  zero-mean complex Gaussian, with arbitrary (but fixed) $(\nt \times \nt)$ covariance matrix ${\bf Q}=\EX{\bf x  x ^{\dag} }$ so that $\tr{\bf Q}=\Pt$.
Then, the function \eqref{eq:defCgeneric} can be used to express the ergodic mutual information in the following cases  \cite{ChiWinZan:J03,SmiRoySha:J03,ShiWinLeeChi:J05}:
\begin{enumerate}
\item the MIMO-$(\nt, \nr)$ channel with no correlation at the receiver (${\bf \Psi}_{\text{R}}={\bf I}$), covariance matrix at the transmitter side ${\bf \Psi}_{\text{T}}$,  and transmit covariance matrix ${\bf Q}$. 

In this case the mutual information is  $\mathcal{C}_{\mathrm{SU}}\left(\nt, \nr,{\bf \Phi}\right)$ with ${\bf \Phi}=({1}/{\sigma^2}) {\bf \Psi}_{\text{T}} {\bf Q}$.  If also ${\bf \Psi}_{\text{T}}={\bf I}$, we have ${\bf \Phi}=({1}/{\sigma^2}) {\bf Q}$ and therefore $\tr{\bf \Phi}=\Pt/\sigma^2$.

\item the MIMO-$(\nt, \nr)$ channel with no correlation at the transmitter (${\bf \Psi}_{\text{T}}={\bf I}$), covariance matrix at the receiver side ${\bf \Psi}_{\text{R}}$, and equal power allocation ${\bf Q}=\Pt/\nt {\bf I}$. 

In this case the capacity is  $\mathcal{C}_{\mathrm{SU}}\left(\nr, \nt,{\bf \Phi}\right)$ with ${\bf \Phi}=({\Pt}/{\nt \sigma^2}) {\bf \Psi}_{\text{R}}$, giving $\tr{\bf \Phi}=({\Pt}/{\sigma^2}) ({\nr}/{\nt})$, in accordance to \cite[Theorem  1]{ChiWinZan:J03}.

\end{enumerate}

In both cases $\Pt/\sigma^2$ represents the \ac{SNR} per receiving antenna. 
}


%

By indicating \mynote{with $\nmin=\min(n,p)$ and}  with $f_{\bl{}} (\cdot, \ldots, \cdot)$ the joint
\ac{p.d.f.} of the (real) ordered \mynote{non-zero} eigenvalues $\lambda_1 \geq
\lambda_2 \geq \ldots \geq \lambda_{\nmin} > 0$ of the $(p
\times p)$ random matrix ${\bf W}={\H {\bf \Phi} \H^{\dag} }$, 
we can write:
\begin{eqnarray}\label{eq:defCmediaeigenv2}
\mathcal{C}_{\mathrm{SU}}\left(n, p, {\bf
\Phi}\right)&=&\EX{\sum_{i=1}^{\nmin}\log\left(1+\lambda_i\right)}
\nonumber
\\
&=&\int \cdots \int_{\Do} f_{\bl{}} (x_{1}, \ldots,
x_{\nmin})\sum_{i=1}^{\nmin}\log\left(1+x_i\right)d\mathbf{x}
\end{eqnarray}
where the multiple integral is over the domain $\Do=\left\{ \infty
> x_1 \geq x_2 \geq \ldots \geq x_{\nmin} > 0 \right\}$ and
$d{\bf x}=dx_1 \, dx_2 \cdots dx_{\nmin}$.

The nested integral in $\eqref{eq:defCmediaeigenv2}$ can be
evaluated using the results from previous sections 
 and Appendix~\ref{app:identitymult}, leading to the following
Theorem.

\begin{theorem}\label{th:meancap}

The ergodic \mynote{mutual information} of a MIMO Rayleigh fading
channel with CSI at the receiver only and one-sided correlation matrix ${\bf \Phi}$ having eigenvalues of
arbitrary multiplicities is given by
%
\begin{equation}\label{eq:meanCmimo}
\mathcal{C}_{\mathrm{SU}}\left(n, p, {\bf \Phi}\right)=K
\sum_{k=1}^{\nmin}\det\left({\bf R}^{(k)}\right).
\end{equation}
\medskip
In the previous equation $\nmin=\min(n,p)$,  
%
%
the matrix ${\bf R}^{(k)}$ has elements
%
\begin{equation}\label{eq:meanCmimoPsi}
r^{(k)}_{i,j}=\begin{cases}
    (-1)^{d_i}\int_0^{\infty}x^{p-\nmin+j-1+d_i} e^{- x \, \mu_{(e_i)}}  dx & j=1, \ldots, \nmin, \, j \neq k \\

    (-1)^{d_i}\int_0^{\infty}x^{p-\nmin+j-1+d_i} e^{- x \, \mu_{(e_i)}} \log\left(1+x\right) dx & j=1, \ldots, \nmin, \, j=k \\

    \, [ n-j ]_{d_i} \,\,\mu_{(e_i)}^{n-d_i-j} & j=\nmin+1, \ldots, n  
  \end{cases}
\end{equation}
%
and  $[a]_k$, $e_i$, $d_i$, $K$ are defined as in Lemma \ref{lemma: pdfquadraticforms}, where  $\mu_{(1)} > \mu_{(2)} \ldots > \mu_{(L)}$ are the $L$ distinct eigenvalues of ${\bf \Phi}^{-1}$, with corresponding multiplicities $m_1, \ldots , m_L$.


\end{theorem}
\begin{proof}


In Section~\ref{sec:quadforms} it is shown that the joint
\ac{p.d.f.} of the ordered eigenvalues of $\mathbf{W}$ can be
written as \eqref{eq:jpdfgeneralquadraticform}, where the elements of ${\bf V(x)}, {\bf \tilde{G}(x,\bm \mu)}$
are real functions of $x_{1}, \ldots, x_{\nmin}$. Thus, by using 
Appendix~\ref{app:identitymult}, the multiple integral in
\eqref{eq:defCmediaeigenv2} reduces to
 \eqref{eq:meanCmimo}.
\end{proof}
%
%

Note that the integral in \eqref{eq:meanCmimoPsi} can be evaluated easily with standard numerical techniques; however, the integral can be further simplified, using the identities
$\int_{0}^{\infty}x^m e^{-x\mu} dx= m! / \mu^{m+1}$, and
$\int_{0}^{\infty} x^m e^{-x\mu} \ln(1+x)dx = m! \, e^{\mu}
\sum_{i=0}^{m}\Gamma(i-m,\mu)/\mu^{i+1}$, where
$\Gamma(\cdot,\cdot)$ is the incomplete Gamma function. 

\mynote{
Theorem \ref{th:meancap} gives, in a unified way, the exact mutual information for MIMO systems, encompassing the cases of $\nr\geq\nt$ and $\nt\geq\nr$ with arbitrary correlation at the transmitter or at the receiver, avoiding the need for Monte Carlo evaluation.
}
The application of the results in Sections~\ref{sec:appendixhypfun}-\ref{sec:unifMIMOcap} enables a unified analysis  for MIMO systems, which allow the generalization \mynote{for ergodic and outage capacity}  \cite{ChiWinZan:J03,SmiRoySha:J03, ShiWinLeeChi:J05,McKCol:05}, for optimum combining multiple antenna systems \cite{ChiWinZanMalWint:J03, ChiWinZan:J03a}, for MIMO-MMSE systems \cite{ZanChiWin:J05}, for MIMO relay networks \cite{WanZhaHos:05,BolNabOymPau:06}, as well as for multiuser MIMO systems and for distributed MIMO systems, accounting arbitrary covariance matrices.
\mynote{For example, after the first derivation of the hypergeometric functions of matrices with non-distinct eigenvalues in \cite{ChiWinShi:C06}, other applications to multiple antenna systems have appeared in \cite{JinMcKGaoCol:08,ShiWin:08,KanKwaPraStu:08,McKZanColChi:J09,ZanChiWin:J09}. }


\section{Numerical Results} \label{sec:MIMOCCIcap}
%



\mynote{Let us first apply Theorem~\ref{th:meancap} to the 
analysis of a single-user MIMO system with
unequal power levels among the transmitting antennas. 
%
}
Figure \ref{fig:ECnt6nr3nq2} shows the ergodic \mynote{mutual information}\footnote{For the numerical results we use the base 2 of logarithm in all formulas, giving a mutual information in bits/s/Hz.}  of a
MIMO-$(6,3)$ Rayleigh channel, where the relative transmitted
power levels are $\{1+\Delta, 1+\Delta, 1+\Delta, 1-\Delta,
1-\Delta, 1-\Delta \}$. 
The particular cases $\Delta=0$ and
$\Delta=1$ are equivalent to the equal power levels over $6$ and
$3$ transmitting antennas, respectively.
This figure shows how the capacity decreases as $\Delta$ increases
from $0$ to $1$, \mynote{with a behavior in accordance to analysis based on majorization theory \cite{MarOlk:B79}}.

%
%

\mynote{
As another example of application, we evaluate the performance of \ac{MIMO} relay networks in Rayleigh fading\cite{WanZhaHos:05,BolNabOymPau:06}. For such networks 
the network capacity is upper bounded by \cite[eq. (5)]{BolNabOymPau:06}, which can be easily put in the form $C_u=\frac{1}{2} \EXs{\mathbf{H}}{\log\det\left({\bf I}+ {\H {\bf \Phi} \H^{\dag} }\right)}$, and evaluated in closed form by Theorem~\ref{th:meancap}.
In Fig. \ref{fig:MIMOrelay} we report the exact $C_u$ as obtained from Theorem~\ref{th:meancap}, compared with the Jensen's inequality \cite[Theorem 1]{BolNabOymPau:06}. The figure has been obtained for a source node with $4$ antennas, $5$ relays each equipped with $2$ antennas, as a function of the total equivalent \ac{SNR} here defined as $\textsf{SNR}=\tr{{\bf \Phi}}$. We assume, for the $5$ relays, that the received power is distributed proportionally to the weights $\{1,2,5,10,20\}$.
It can observed that the results based on the Jensen's inequality can be overly optimistic.
}


As a third example of application we evaluate, using \eqref{eq:meanCapMIMOCCI} together with
Theorem~\ref{th:meancap}, the exact expression of the
ergodic \mynote{mutual information} of MIMO systems in the presence of multiple MIMO
 interferers in Rayleigh fading. In particular, the
eigenvalues to be used in Theorem~\ref{th:meancap} are given by
$\mu_{(i)}=1/\r_i=\sigma^2 \Nt_i/\Pt_i$, allowing an easy analysis
for several scenarios.
We define the average \ac{SNR} per receiving antenna as $\textsf{SNR}=\Pt_0/\sigma^2$ giving $\r_0=\textsf{SNR}/\Nt_0$, and the \ac{SIR} as
$\textsf{SIR}=\Pt_0/\sum_{i\geq 1}\Pt_i$.\footnote{We recall that, 
with our normalization on the channel gains, the mean received
power from user $i$ is $\Pt_i$, and our definition of $\mathsf{SIR}$ account for the
$total$ interference power.} 
%
%
Fig.~\ref{figuraMIMOCCISNR10dbNR6.eps} shows the
ergodic \mynote{mutual information} for a MIMO-$(6,6)$ system as a function of the
\ac{SIR}, in the presence of one MIMO cochannel interferer having
$\Nt_1$ equal power transmitting antennas. 
%
It can be noted that the capacity decreases with the increase in
the number of interfering antenna elements, tending to the curve
obtained by using the Gaussian approximation.\footnote{With
Gaussian approximation the performance is evaluated as if
interference were absent, except the overall noise power is set to
$\sigma^2+\sum_{i\geq 1}\Pt_i$, giving a signal-to-interference-plus-noise ratio $ \textsf{SINR}
=\left(\frac{1}{\textsf{SNR}}+\frac{1}{\textsf{SIR}}\right)^{-1}
$ .}
Despite the fact that the received vector $\mathbf{y}$ in
\eqref{eq:y} is Gaussian conditioned on the channel matrices, and that the elements of 
$\mathbf{H}_k$ are Gaussian, 
approximating the cumulative interference as a spatially white
complex Gaussian vector is pessimistic 
 for analyzing MIMO systems in the presence of interference, unless the number of transmitting antenna of the
interferer is large compared with that of the desired user. This is
because the Gaussian approximation implicitly assumes that the
receiver does not exploit the \ac{CSI} of the interferers (single-user receiver), whereas the exact capacity accounts for the
knowledge of all \ac{CSI} at the receiver.
In the same figure we also report, using circles, the capacity 
of a single-user MIMO-$(\Nt_0,\Nr-\Nt_1)$  for $\Nr > \Nt_1$. It
can be observed that the capacity of the MIMO-$(\Nt_0,\Nr)$ in the
presence of $\Nt_1$ interfering antenna elements approaches
asymptotically, for large interference power, 
 to 
a floor given by the capacity of a single-user
MIMO-$(\Nt_0,\Nr-\Nt_1)$ system. This behavior can be thought of
as using $\Nt_1$ \ac{DoF} at the receiver to null the
interference in a small \ac{SIR} regime.
On the other hand, when $\Nr\leq\Nt_1$ the capacity approaches to
zero for small \ac{SIR}.
This is due to the limited \ac{DoF} at the receiver (related to
the number $\Nr$ of receiving antenna elements) that prevents
mitigating all interfering signals (one from each antenna
elements
) while, at the
same time, processing the $\Nt_0$ useful parallel streams, \mynote{as previously observed for multiple antenna systems with optimum combining} 
\cite{WintSalGit:94,ChiWinZanMalWint:J03,ChiWinZan:J03a}.


Finally, in Fig.~\ref{figuraMIMOnR6CCISNR10dbONEandTWOandGAUSINTERF.eps} we
consider a MIMO-$(\Nt_0,6)$ system in the presence of one and two MIMO
interferers in the network, each equipped with the same number of
antennas as for the desired user. 
%
We clearly see here two different regions: for small \ac{SIR} the
interference effect is dominant, and it is better for
all users to employ the minimum number of transmitting antennas
(i.e., MIMO-$(3,6)$ for all users), so as to allow the receiver to
mitigate the interfering signals. On the contrary, for large 
\ac{SIR} the channel tends to that of a single-user MIMO system
and it is better to employ the maximum number of transmitting
antennas.
%
%
In the same figure we also report the capacity for
interference-free channels, which represents the asymptotes of the
four curves, 
%
%
%
%
%
as well as the Gaussian 
approximation, which incorrectly indicates that it is always better to use
the largest possible number of transmitting antennas.
%
%

%
\mynote{It can be also verified that,} in a network where all nodes are using the same MIMO-$(n,n)$
systems, larger values of $n$ achieve higher \mynote{mutual information}, for all
values of \ac{SIR} and \ac{SNR}.
\mynote{
Note, however, that when increasing the number of antennas and users, correlation may arise in the channel matrices.
}

\section{Conclusion}\label{sec:conclusion}
\mynote{We have studied MIMO communication systems in the presence of multiple MIMO
interferers and noise.} To this aim, we \mynote{first} generalized the determinant
representations for hypergeometric functions with matrix arguments
to the case where the eigenvalues of the argument matrices have 
arbitrary multiplicities. 
\mynote{
Then, we derived a unified formula for the joint \ac{p.d.f.} of the eigenvalues for central Wishart matrices and Gaussian quadratic forms, allowing arbitrary multiplicities for the covariance matrix eigenvalues. 
} 
These new results enable the analysis of
many scenarios involving MIMO systems.
\mynote{
For example, we derived a unified expression for the ergodic mutual information of MIMO Rayleigh fading channels,  which applies to transmit or receive correlation matrices with eigenvalues of arbitrary multiplicities.}
\mynote{We have shown how to apply the new expressions to MIMO networks,  deriving in closed form the
ergodic mutual information of MIMO systems in the presence of multiple MIMO
interferers.}


\begin{appendices}

\section{Proofs} 
\label{app:proofs}


\subsection{Proof of Lemma  \ref{lemma: lemma0chiani} } 

%
For ease of notation and without loss of generality we consider
the case $K=1$, where the application of the lemma leads to
\eqref{eq:theorem0}. For the proof we proceed by induction.
%
%
%
First, the result in \eqref{eq:theorem0} is obvious for
$L=1$, since in this case \eqref{eq:theorem0} coincides with
\eqref{eq:defPw1wm}.
%
%
%
Then, we must show that if \eqref{eq:theorem0} is true for any 
 $L$ then it is also true for $L+1$. So, assuming that
\eqref{eq:theorem0} holds for $L$, we must find
$$\lim_{w_{L+1} \rightarrow w_{L}} \breve{P}(w_1,\ldots,w_m).$$



In this regard note that, with $w_1=w_2=\cdots=w_L$ the product
$\prod_{i<j, w_i \neq w_j}(w_i-w_j)$ in \eqref{eq:theorem0}
contains exactly $L$ factors with value $\epsilon \teq
w_L-w_{L+1}$. Thus, by rewriting $w_{L+1}=w_{L}-\epsilon$ we have

\bigskip
 \lefteqn{\lim_{w_{L+1} \rightarrow w_{L}} \breve{P}(w_1,\ldots,w_m)
=\frac{1}{\prod_{i<j, w_i \neq w_j, j\neq L+1}(w_i-w_j)
\prod_{i=1}^{L-1} i!}\cdot} 
\begin{equation}\label{eq:theorem0p2}
\cdot \lim_{\epsilon\rightarrow 0} \frac{1}{\epsilon^L} \left|\begin{array}{llllllll} 
    f^{(L-1)}_1(w_L) & \cdots & f_1(w_L) & f_1(w_{L}-\epsilon) & \cdots & f_1(w_m) \\
                      &       & &            &        &         \\
    \vdots            & & \vdots  &           & \cdots & \vdots  \\
                      &        & &            &        &         \\
    f^{(L-1)}_m(w_L) & \cdots & f_m(w_L) & f_m(w_{L}-\epsilon) & \cdots & f_m(w_m) \\
  \end{array}\right|\,.
\end{equation}
%
We can now apply the Taylor expansion to the functions
\begin{equation}\label{eq:taylor}
  f_i(w-\epsilon)=\sum_{n=0}^{L} f^{(n)}_i(w) \frac{(-\epsilon)^n}{n!}+O(\epsilon^{L+1}),
\end{equation}
where $O(\epsilon)$ denotes omitted terms of order $\epsilon$. We
also know from basic algebra that, seen as a function of a column
with the others fixed, the determinant is a linear function of the
entries in the given column, as is clear for example from the
Laplace expansion. Therefore, we have

\bigskip

{
\lefteqn{\lim_{w_{L+1} \rightarrow w_{L}}
\breve{P}(w_1,\ldots,w_m) = }

{
$$=\frac{1}{ \prod_{i<j, w_i \neq w_j, j\neq
L+1}(w_i-w_j) \prod_{i=1}^{L-1} i!}\cdot \lim_{\epsilon
\rightarrow 0}\left(O(\epsilon)+\frac{1}{\epsilon^L}
\sum_{n=0}^{L} \frac{(-\epsilon)^n}{n!} \right.$$}}

{
\begin{equation}\label{eq:theorem0p3}
\left. \cdot   \left|\begin{array}{llllllll} 
    f^{(L-1)}_1(w_L) & \cdots & f_1(w_L) & f^{(n)}_1(w_{L}) & \cdots & f_1(w_m) \\
                      &        & &            &        &         \\
    \vdots            & &\vdots  &           & \cdots & \vdots  \\
                      &        & &            &        &         \\
    f^{(L-1)}_m(w_L) & \cdots & f_m(w_L) & f^{(n)}_m(w_{L}) & \cdots & f_m(w_m) \\
  \end{array}\right|\right)\,.
\end{equation}
}

In the summation above the determinants for $n=0,\ldots,L-1$ are
zero since there are coincident columns. So, in the limit for
${\epsilon \rightarrow 0}$ only the term of grade $L$ remains.

%
%
%
%

By simplifying  and reordering the first $L+1$ columns of the
matrix in \eqref{eq:theorem0p3}, with a cyclic permutation having
sign equal to $(-1)^L$, we finally have

\bigskip

\lefteqn{\lim_{w_L \rightarrow w_{L+1}} \breve{P}(w_1,\ldots,w_m)
=\frac{1}{ \prod_{i<j, w_i \neq w_j}(w_i-w_j) \prod_{i=1}^{L}
i!}\cdot}

{
\begin{equation}\label{eq:theorem0p5}
\cdot  \left|\begin{array}{llllllll} 
    f^{(L)}_1(w_{L+1}) & \cdots & f_1(w_{L+1}) & f_1(w_{L+2}) & \cdots & f_1(w_m) \\
                      &       & &             &        &         \\
    \vdots            & \vdots &&            & \cdots & \vdots  \\
                      &       & &             &        &         \\
    f^{(L)}_m(w_{L+1}) & \cdots & f_m(w_{L+1}) & f_m(w_{L+2}) & \cdots & f_m(w_m) \\
  \end{array}\right|\,
\end{equation}
}

\normalsize

which is again in the form of \eqref{eq:theorem0}. This concludes
the proof by induction of Lemma~\ref{lemma: lemma0chiani} for
$w_1=\cdots=w_L$.

The extension to different $K$ and more groups of coincident
arguments is straightforward.

\subsection{Proof of Lemma \ref{lemma: lemma3chiani}.}
The derivatives of the hypergeometric function of scalar arguments can be expressed as
%

{
$$ \frac{d^n}{d z^n} {}_p\td{F}_q\left(a_1,\ldots,a_p; b_1,\ldots,b_q; z\right)=
\frac{(a_1)_n \cdots (a_p)_n}{(b_1)_n\cdots (b_q)_n}
{}_p\td{F}_q\left(a_1+n,\ldots,a_p+n; b_1+n,\ldots,b_q+n;
z\right).
$$
}
Using this result in Lemma~\ref{lemma: lemma0chiani} and
\eqref{eq:pFqXYkhatri} with
$$f_i(w)={}_p\td{F}_q\left(\tilde{a}_1,\ldots,\tilde{a}_p; \tilde{b}_1,\ldots,\tilde{b}_q; \lambda_i w
\right)$$
gives Lemma~\ref{lemma: lemma3chiani}.


\subsection{Proof of Lemma \ref{lemma: pdfquadraticforms}}
%
\mynote{Here, based on Section~\ref{sec:appendixhypfun}, we prove Lemma \ref{lemma: pdfquadraticforms}} 
 concerning the eigenvalues distribution of Gaussian
quadratic forms. The problem is related to the distribution of
random matrices of the form
$\mathbf{W}=\mathbf{H}\mathbf{\Phi}\mathbf{H}^{\dag}$, where
$\mathbf{H}$ is a Gaussian $(p \times n)$ matrix with
uncorrelated entries and $\mathbf{\Phi}$ is a $(n \times n)$
positive definite matrix that represents the covariance matrix of
the channel. The eigenvalues distribution has been studied for the
two possible cases $n\geq p$ and $p\geq n$ in
\cite{ChiWinZan:J03,SmiRoySha:J03}, assuming a covariance matrix
$\mathbf{\Phi}$ with distinct eigenvalues (i.e., unit 
multiplicity). We here generalize the results to matrices
$\mathbf{\Phi}$ with arbitrary eigenvalue multiplicities.

Let us first recall the distributions for the case of covariance
matrix with distinct eigenvalues.

\subsubsection{Correlation on the shortest side - distinct eigenvalues}
%
The case $p \geq n$ has been analyzed in \cite{ChiWinZan:J03},
where it is shown that the joint \ac{p.d.f.} of the (real) ordered
eigenvalues $\lambda_1 \geq \lambda_2 \geq \ldots \geq
\lambda_{n}$ of $\mathbf{W}$ is
\begin{equation} \label{eq:jpdfSchianiwinzan}
f_{\bl{}} (x_{1}, \ldots, x_{n})
    = \frac{1}{\Gamma_{(n)}(p)} \, \frac{\prod_{i=1}^{n}\mu_i^p}{\prod_{i<j}\left(\mu_i-\mu_j\right)}
    \left|{\bf V(x)}\right| \,
        \left|{\bf G}({\bf x,\bm \mu})\right| \prod_{j=1}^{n} x_j^{p-n}
\end{equation}
where $\mu_i$ are the $n$ distinct eigenvalues of
$\mathbf{\Phi}^{-1}$, ${\bf V(x)}$ is the ($n\times n$) Vandermonde
matrix with elements $v_{i,j}=x^{i-1}_j$ and where ${\bf G}({\bf
x,\bm\mu})$ is a ($n\times n$) matrix with elements $g_{i,j}=e^{-\mu_i x_j
}$.
\subsubsection{Correlation on the largest side - distinct eigenvalues}
%
We here \mynote{briefly} derive the joint \ac{p.d.f.} for the eigenvalues
of $\mathbf{W}$ when $\mathbf{\Phi}$ has all distinct eigenvalues
and $n\geq p$, 
based on the results in
Section~\ref{sec:appendixhypfun}. Note that this case has been
analyzed also in \cite{SmiRoySha:J03} by following a different
approach.

First we recall that, given a $(p \times n)$ random matrix
$\mathbf{H}$ with $n\geq p $ and \ac{p.d.f.}
\begin{equation}
\label{eq:pdfX} \pi^{-p n}
e^{-tr \, \mathbf{H} \mathbf{H}^\dag}
\end{equation}
the \ac{p.d.f.} of the $(p \times p)$ quadratic form
\begin{equation}
\label{eq:defXLX}
\mathbf{W}=\mathbf{H}\mathbf{\Phi}\mathbf{H}^{\dag} ,
\end{equation}
where the $(n \times n)$ matrix $\mathbf{\Phi}$ is positive
definite, is given by \cite{Kha:66, Hay:66}
\begin{equation}
\label{eq:pdfS} f(\mathbf{W})=\frac{\left|\mathbf{W}\right|^{n-p}}{\pi^{(p-1)p/2} \, \Gamma_{(p)}(n) 
\left|\mathbf{\Phi}\right|^{p}} \,\,
{}_0\td{F}_0\left(\mathbf{\Phi}^{-1}, -\mathbf{W}\right).
\end{equation}
%
Then, the joint \ac{p.d.f.} of the (real) ordered eigenvalues
$\lambda_1 \geq \lambda_2 \geq \ldots \geq \lambda_{p}$ of
$\mathbf{W}$ is given by using the results in \cite[eq.
(93)]{Jam:64} as
\begin{equation} \label{eq:jpdfS}
f_{\bl{}} (x_{1}, \ldots, x_{p})
    = K_1 \left|\mathbf{\Phi}\right|^{-p} \,
        {}_0\td{F}_0\left(\mathbf{\Phi}^{-1},-\mathbf{W}\right)
        \left|\mathbf{W}\right|^{n-p}
        \cdot \prod_{i<j}^{p} \left(x_{i}-x_{j}\right)^2\,,
\end{equation}
where
\begin{equation}
K_1 = \frac{1}
        {{\Gamma}_{p}(n){\Gamma}_{p}(p)} \, .
\label{eq:K}
\end{equation}
%
%
Note that in \eqref{eq:jpdfS} the two matrices
$\mathbf{\Phi}^{-1}$ and $\mathbf{W}$ are of dimensions $(n \times n)$
and $(p \times p)$, respectively. So, in \eqref{eq:jpdfS} we
evaluate ${}_0\td{F}_0\left(\mathbf{\Phi}^{-1},{\bf B}\right)$
where ${\bf B}=-\mathbf{W}\oplus 0\cdot \mathbf{I}_p$ is obtained
by adding $n-p$ zero eigenvalues to $-\mathbf{W}$
\cite{SmiRoySha:J03}. 

Differently from the previous literature, we
can now directly use Corollary \ref{corollary:lemma1chiani} and get immediately the 
 joint \ac{p.d.f.} of the ordered
eigenvalues of the $(p\times p)$ matrix $\mathbf{W}$ when $n\geq p$ as:
\begin{equation} \label{eq:jpdfSchiani}
f_{\bl{}} (x_{1}, \ldots, x_{p})
    = \frac{(-1)^{p(n-p)}}{\Gamma_{(p)}(p)} \, \frac{\prod_{i=1}^{n}\mu_i^p}{\prod_{i<j}\left(\mu_i-\mu_j\right)}
    \left|{\bf V(x)}\right| \,
        \left|{\bf G(x,\bm \mu)}\right|
\end{equation}
where $\mu_i$ are the eigenvalues of $\mathbf{\Phi}^{-1}$, all of
multiplicity one here, ${\bf V(x)}$ is the ($p\times p$)
Vandermonde matrix, and the ($m\times m$) matrix ${\bf G(x,\bm
\mu)}$ has elements as follows
\begin{equation} g_{i,j}=\left\{
  \begin{array}{ll}
    e^{-\mu_i x_j} & \qquad j=1,\ldots,p \\
    \mu_i^{n-j} & \qquad j=p+1,\ldots,n 
  \end{array}
  \right.
\end{equation}
that is, the matrix $\bf G(x,\bm\mu)$ is
{
\begin{equation}\label{eq:Gxmudistinct} \mb{G(x,\bm\mu)} \teq
  \left[\begin{array}{cccccccc} 
    e^{-\mu_1 {x}_1} & \cdots & e^{-\mu_1 {x}_{p}} & \mu_1^{n-p-1}& \mu_1^{n-p-2}& \cdots & \mu_1 & 1  \\
    e^{-\mu_2 {x}_1} & \cdots & e^{-\mu_2 {x}_{p}} & \mu_2^{n-p-1}& \mu_2^{n-p-2}& \cdots & \mu_2 & 1  \\
                      &        &                       &                &                &        &           &    \\
    \vdots            &        & \vdots                & \vdots         & \vdots         & \cdots & \vdots    & \\
                      &        &                       &                &                &        &           &    \\
    e^{-\mu_n {x}_1} & \cdots & e^{-\mu_n {x}_{p}} & \mu_n^{n-p-1}& \mu_n^{n-p-2}& \cdots & \mu_n & 1  \\
  \end{array}\right]\,
=
  \left[\begin{array}{c} 
    \mathbf{g}({\bf x}, \mu_1) \\
    \mathbf{g}({\bf x}, \mu_2) \\
\vdots \\
    \mathbf{g}({\bf x}, \mu_n) \\
  \end{array}\right]\,
  .
\end{equation}
}


\subsubsection{Generalization to covariance matrix with arbitrary eigenvalues}

Note that \eqref{eq:jpdfSchianiwinzan} and
\eqref{eq:jpdfSchiani} are only valid for covariance matrices with all
distinct eigenvalues (multiplicity one). So, we must now
generalize these expressions to the case of interest, i.e.,
eigenvalues $\mu_i$ with arbitrary multiplicities. This step is
possible by using Lemma~\ref{lemma: lemma0chiani}.

In fact, we note that in both \eqref{eq:jpdfSchianiwinzan} and
\eqref{eq:jpdfSchiani} we have a ratio of the form
\begin{equation} \label{eq:ratio}
    \frac{\left|{\bf
    G(x,\bm\mu)}\right|}{\prod_{i<j}\left(\mu_i-\mu_j\right)}.
\end{equation}
By using Lemma~\ref{lemma: lemma0chiani}, for each eigenvalue with
multiplicity $m_i$ we must replace the rows of ${\bf G(x,\bm \mu)}$
with their successive derivatives with respect to the eigenvalue,
and divide by $\Gamma_{(m_i)}(m_i)$, obtaining
\begin{equation}\label{eq:Gsulambdamultiplicity}
\frac{\left|{\bf G(x,\bm \mu)}\right|}{\prod_{i<j}\left(\mu_i-\mu_j\right)}\rightarrow
\frac{1}{\prod_i\Gamma_{(m_i)}(m_i) \, \prod_{i<j}\left(\mu_{(i)}-\mu_{(j)}\right)^{m_i m_j}}  \det \left[\begin{array}{c} 
    \mathbf{g}^{(m_1-1)}({\bf x}, \mu_{(1)}) \\
\vdots \\
    \mathbf{g}^{(1)}({\bf x}, \mu_{(1)}) \\
    \mathbf{g}^{(0)}({\bf x}, \mu_{(1)}) \\
\vdots \\
    \mathbf{g}^{(m_L-1)}({\bf x}, \mu_{(L)}) \\
\vdots \\
    \mathbf{g}^{(1)}({\bf x}, \mu_{(L)}) \\
    \mathbf{g}^{(0)}({\bf x}, \mu_{(L)}) \\
  \end{array}\right]\,
\end{equation}
where the row vector $\mathbf{g}^{(l)}({\bf x}, \mu)$ is the $\ith{l}$
derivative of the row $\mathbf{g}({\bf x},\mu)$ in \eqref{eq:jpdfSchianiwinzan} or \eqref{eq:Gxmudistinct}.
The $\ith{j}$ element of $\mathbf{g}^{(l)}({\bf x},\mu)$ is so derived to be
\begin{equation} g^{(l)}_{j}=g^{(l)}_{j}(\mu)=\left\{
  \begin{array}{ll}
    \left(-x_j \right)^l e^{-\mu x_j} & \qquad j=1,\ldots,p \\
    \left[n-j\right]_l \mu^{n-j-l}& \qquad j=p+1,\ldots,n .\\
  \end{array}
  \right. 
\end{equation}
%

\mynote{The relation between the row index, $i$, and the derivative order, $l$,  can be established by introducing the function $e_i$ indicating the eigenvalue $\mu_{(e_i)} \in \left\{\mu_{(1)}, \ldots, \mu_{(L)}\right\}$ to be used in row $i$ of the matrix in the RHS of  \eqref{eq:Gsulambdamultiplicity}. It is easy to verify that $e_i$ is the unique integer such that
$$m_1+\ldots+m_{e_i-1}< i \leq m_1 +\ldots+m_{e_i}. $$
Then, the derivative order for the row $i$ is $l=d_i$, where
$$d_i=\sum_{k=1}^{e_i}m_k -i .$$
Thus, the generic element of the matrix in the RHS of \eqref{eq:Gsulambdamultiplicity} is $g^{(d_i)}_{j}(\mu_{(e_i)})$.

\noindent Combining \eqref{eq:jpdfSchianiwinzan}, \eqref{eq:jpdfSchiani} and \eqref{eq:Gsulambdamultiplicity} we have Lemma  \ref{lemma: pdfquadraticforms}.

}


\section{An identity on multiple integrals involving determinants}
\label{app:identitymult}
\begin{theorem}\label{th:meansum}
Given an arbitrary $p \times p$ matrix ${\bf \Phi}\left({\bf
x}\right)$ with $\ith{ij}$ elements $\Phi_i(x_j)$, an arbitrary $n
\times n$  matrix ${\bf \Psi}\left({\bf x}\right)$, $n\geq p$,
with elements
$$
    \left\{\begin{array}{ll}
    \Psi_i(x_j) & \qquad j=1,\ldots,p \\
    \Psi_{i,j} & \qquad j=p+1,\ldots,n 
  \end{array}
  \right. 
$$
and two arbitrary functions $\xi(\cdot)$ and $\td{\xi}(\cdot)$ the
following identity holds:
\begin{eqnarray}\label{eq:appdxsumunord}
 \hspace{-2.5cm} \int \ldots \int_{\Do}
    \left| {\bf \Phi}\left({\bf x}\right)\right| \cdot
    \left| {\bf \Psi}\left({\bf x}\right)\right|
    \prod_{m=1}^{p} \xi(x_m) \sum_{i=1}^{p} \td{\xi}(x_i) d{\bf x} && \nonumber \\
  &=& \sum_{k=1}^{p} \det \left(\left\{c^{(k)}_{i,j}\right\}_{i,j=1\ldots,n}\right)
\end{eqnarray}
where the multiple integral is over the domain $\Do=\left\{b\geq
x_1 \geq x_2 \geq \ldots \geq x_p \geq a \right\}$ , 
$$c^{(k)}_{i,j}=
    \left\{\begin{array}{ll}
    \int_a^b \Phi_i(x) \Psi_j(x) \; \xi(x) U_{k,j}\left(\td{\xi}(x)\right) dx & j=1,\ldots,p \\
    \Psi_{i,j} & j=p+1,\ldots,n
  \end{array}
  \right. 
$$
and the function
$U_{k,j}(x)$ is defined by
\begin{equation} \label{eq:Ukj}
  U_{k,j}(x) \teq \begin{cases}
                    x & \quad \text{if} \quad k=j \\
                    1 & \quad \text{if} \quad k\neq j .
                \end{cases} \, 
\end{equation}

\end{theorem}

\bigskip

\begin{proof}
As this theorem is an extension of \cite[Theorem 3]{ChiWinZan:J03}, it is sufficient for the proof to follow the same steps reported there.
\end{proof}

\end{appendices}

\section*{Acknowledgments}
The authors would like to thank M. Cicognani, the Associate Editor A. Goldsmith and the anonymous reviewers for their useful comments.


\bibliographystyle{IEEEtran}

\begin{thebibliography}{10}
\providecommand{\url}[1]{#1}
\csname url@samestyle\endcsname
\providecommand{\newblock}{\relax}
\providecommand{\bibinfo}[2]{#2}
\providecommand{\BIBentrySTDinterwordspacing}{\spaceskip=0pt\relax}
\providecommand{\BIBentryALTinterwordstretchfactor}{4}
\providecommand{\BIBentryALTinterwordspacing}{\spaceskip=\fontdimen2\font plus
\BIBentryALTinterwordstretchfactor\fontdimen3\font minus
  \fontdimen4\font\relax}
\providecommand{\BIBforeignlanguage}[2]{{%
\expandafter\ifx\csname l@#1\endcsname\relax
\typeout{** WARNING: IEEEtran.bst: No hyphenation pattern has been}%
\typeout{** loaded for the language `#1'. Using the pattern for}%
\typeout{** the default language instead.}%
\else
\language=\csname l@#1\endcsname
\fi
#2}}
\providecommand{\BIBdecl}{\relax}
\BIBdecl

\bibitem{Wint:87}
J.~H. Winters, ``On the capacity of radio communication systems with diversity
  in {Rayleigh} fading environment,'' \emph{IEEE J. Select. Areas Commun.},
  vol. SAC-5, no.~5, pp. 871--878, Jun. 1987.

\bibitem{WintSalGit:94}
J.~H. Winters, J.~Salz, and R.~D. Gitlin, ``The impact of antenna diversity on
  the capacity of wireless communication system,'' \emph{IEEE Trans. Commun.},
  vol.~42, no. 2/3/4, pp. 1740--1751, Feb./Mar./Apr. 1994.

\bibitem{Fos:96}
G.~J. Foschini, ``Layered space-time architecture for wireless communication a
  fading environment when using multiple antennas,'' \emph{Bell Labs Tech. J.},
  vol.~1, no.~2, pp. 41--59, Autumn 1996.

\bibitem{Tel:99}
E.~Telatar, ``Capacity of multi-antenna {G}aussian channels,'' \emph{Europ.
  Trans. on Telecomm.}, vol.~10, no.~6, pp. 585--595, Nov.-Dec. 1999.

\bibitem{Chi:02}
M.~Chiani, ``Evaluating the capacity distribution of {MIMO} {R}ayleigh fading
  channels,'' in \emph{Proc. IEEE Int. Symp. on Advances in Wireless Commun.},
  Victoria, CANADA, Sep. 2002, {\bf Invited Paper}.

\bibitem{ChiWinZan:J03}
M.~Chiani, M.~Z. Win, and A.~Zanella, ``On the capacity of spatially correlated
  {MIMO} {Rayleigh} fading channels,'' \emph{IEEE Trans. Inform. Theory},
  vol.~49, no.~10, pp. 2363--2371, Oct. 2003.

\bibitem{SmiRoySha:J03}
P.~J. Smith, S.~Roy, and M.~Shafi, ``Capacity of {MIMO} systems with
  semicorrelated flat fading,'' \emph{IEEE Trans. Inform. Theory}, vol.~49,
  no.~10, pp. 2781--2788, Oct. 2003.

\bibitem{ShiWinLeeChi:J05}
H.~Shin, M.~Win, J.~H. Lee, and M.~Chiani, ``On the capacity of doubly
  correlated {MIMO} channels,'' \emph{{IEEE} Trans. Wireless Commun.}, vol.~5,
  no.~8, pp. 2253--2266, Aug. 2006.

\bibitem{GioSmiShaChi:J03}
A.~Giorgetti, P.~J. Smith, M.~Shafi, and M.~Chiani, ``{MIMO} capacity, level
  crossing rates and fades: The impact of spatial/temporal channel
  correlation,'' \emph{KICS/IEEE Int. Journal of Communications and Networks},
  vol.~5, no.~2, pp. 104 -- 115, Jun. 2003, (special issue on Coding and Signal
  Processing for {MIMO} systems).

\bibitem{CatDriGre:00}
S.~Catreux, P.~F. Driessen, and L.~J. Greenstein, ``Simulation results for an
  interference-limited multiple-input multiple-output cellular system,''
  \emph{IEEE Commun. Lett.}, vol.~4, no.~11, pp. 334--336, Nov. 2000.

\bibitem{BluWintSol:02}
R.~S. Blum, J.~H. Winters, and N.~Sollenberger, ``On the capacity of cellular
  systems with {MIMO},'' \emph{{IEEE} Commun. Lett.}, vol.~6, no.~6, pp.
  242--244, Jun. 2002.

\bibitem{Blu:03}
R.~S. Blum, ``{MIMO} capacity with interference,'' \emph{IEEE J. Select. Areas
  Commun.}, vol.~21, no.~5, pp. 793--801, Jun. 2003.

\bibitem{MouSimSen:03}
A.~L. Moustakas, S.~H. Simon, and A.~M. Sengupta, ``{MIMO} capacity through
  correlated channels in the presence of correlated interferers and noise: A
  (not so) large n analysis,'' \emph{IEEE Trans. Inform. Theory}, vol.~49,
  no.~10, pp. 2545--2561, Oct. 2003.

\bibitem{DaiPoo:03}
H.~Dai and H.~V. Poor, ``Asymptotic spectral efficiency of multicell {MIMO}
  systems with frequency-flat fading,'' \emph{Signal Processing, IEEE
  Transactions on}, vol.~51, no.~11, pp. 2976--2988, Nov 2003.

\bibitem{DaiMolPoo:04}
H.~Dai, A.~F. Molisch, and H.~V. Poor, ``Downlink capacity of
  interference-limited {MIMO} systems with joint detection,'' \emph{Wireless
  Communications, IEEE Transactions on}, vol.~3, no.~2, pp. 442--453, March
  2004.

\bibitem{LozTulVer:03}
A.~Lozano, A.~Tulino, and S.~Verdu, ``Multiple-antenna capacity in the
  low-power regime,'' \emph{Information Theory, IEEE Transactions on}, vol.~49,
  no.~10, pp. 2527--2544, Oct. 2003.

\bibitem{LozTulVer:05}
------, ``{High-SNR} power offset in multiantenna communication,''
  \emph{Information Theory, IEEE Transactions on}, vol.~51, no.~12, pp.
  4134--4151, Dec. 2005.

\bibitem{JorBoc:04}
E.~A. Jorswieck and H.~Boche, ``Performance analysis of capacity of {MIMO}
  systems under multiuser interference based on worst-case noise behavior,''
  \emph{EURASIP Journal on Wireless Communications and Networking}, vol. 2004,
  no.~2, pp. 273 -- 285, 2004.

\bibitem{TarRie:07}
G.~Taricco and E.~Riegler, ``On the ergodic capacity of the asymptotic
  separately-correlated {Rician} fading {MIMO} channel with interference,'' in
  \emph{Information Theory, 2007. ISIT 2007. IEEE International Symposium on},
  Nice, France, June 2007, pp. 531--535.

\bibitem{HorJoh:B90}
R.~A. Horn and C.~R. Johnson, \emph{Matrix Analysis}, 1st~ed.\hskip 1em plus
  0.5em minus 0.4em\relax Cambridge: Cambridge University Press, 1990.

\bibitem{CovTho:B91}
T.~A. Cover and J.~A. Thomas, \emph{Elements of Information Theory},
  1st~ed.\hskip 1em plus 0.5em minus 0.4em\relax New York, NY, 10158: John
  Wiley \& Sons, Inc., 1991.

\bibitem{Jam:64}
A.~T. James, ``Distributions of matrix variates and latent roots derived from
  normal samples,'' \emph{Ann. Math. Statist.}, vol.~35, pp. 475--501, 1964.

\bibitem{Mui:B82}
R.~J. Muirhead, \emph{Aspects of Multivariate Statistical Theory},
  1st~ed.\hskip 1em plus 0.5em minus 0.4em\relax New York, NY: John Wiley \&
  Sons, Inc, 1982.

\bibitem{Kha:70}
C.~G. Khatri, ``On the moments of traces of two matrices in three situations
  for complex multivariate normal populations,'' \emph{Sankhya, The Indian
  Journal of Statistics, Ser. A}, vol.~32, pp. 65--80, 1970.

\bibitem{GaoSmi:00}
H.~Gao and P.~J. Smith, ``A determinant representation for the distribution of
  quadratic forms in complex normal vectors,'' \emph{Journal of Multivariate
  Analysis}, vol.~73, no.~2, pp. 155 -- 165, 2000.

\bibitem{ChiWinZanMalWint:J03}
M.~Chiani, M.~Z. Win, A.~Zanella, R.~K. Mallik, and J.~H. Winters, ``Bounds and
  approximations for optimum combining of signals in the presence of multiple
  co-channel interferers and thermal noise,'' \emph{IEEE Trans. Commun.},
  vol.~51, no.~2, pp. 296--307, Feb. 2003.

\bibitem{ChiWinZan:J03a}
M.~Chiani, M.~Z. Win, and A.~Zanella, ``Error probability for optimum combining
  of {$M$-ary PSK} signals in the presence of interference and noise,''
  \emph{IEEE Trans. Commun.}, vol.~51, no.~11, pp. 1949--1957, Nov. 2003.

\bibitem{ZanChiWin:J05}
A.~Zanella, M.~Chiani, and M.~Z. Win, ``{MMSE} reception and successive
  interference cancellation for {MIMO} systems with high spectral efficiency,''
  \emph{IEEE Trans. Wireless Commun.}, vol.~4, no.~3, pp. 1244--1253, May 2005.

\bibitem{McKCol:05}
M.~McKay and I.~Collings, ``General capacity bounds for spatially correlated
  {Rician MIMO} channels,'' \emph{Information Theory, IEEE Transactions on},
  vol.~51, no.~9, pp. 3121--3145, Sept. 2005.

\bibitem{ChiWinZan:J05}
M.~Chiani, M.~Z. Win, and A.~Zanella, ``On optimum combining of {$M$-ary PSK}
  signals with unequal-power interferers and noise,'' \emph{IEEE Trans.
  Commun.}, vol.~53, no.~1, pp. 44--47, Jan. 2005.

\bibitem{ZanChiWin:C05}
A.~Zanella, M.~Chiani, and M.~Z. Win, ``Performance of {MIMO} {MRC} in
  correlated {Rayleigh} fading environments,'' in \emph{Proc. IEEE {S}emiannual
  {V}eh. {T}echnol. {C}onf.}, Stockholm, SWEDEN, May 2005.

\bibitem{ZanChiWin:J09}
------, ``On the marginal distribution of the eigenvalues of {W}ishart
  matrices,'' \emph{IEEE Trans. Commun.}, vol.~57, no.~4, pp. 1050--1060, April
  2009.

\bibitem{ChiZan:C08}
M.~Chiani and A.~Zanella, ``Joint distribution of an arbitrary subset of the
  ordered eigenvalues of {Wishart} matrices,'' in \emph{Proc. IEEE Int. Symp.
  on Personal, Indoor and Mobile Radio Commun.}, Cannes, France, Sep. 2008, pp.
  1--6, invited Paper.

\bibitem{WanZhaHos:05}
B.~Wang, J.~Zhang, and A.~Host-Madsen, ``On the capacity of {MIMO} relay
  channels,'' \emph{Information Theory, IEEE Transactions on}, vol.~51, no.~1,
  pp. 29--43, Jan. 2005.

\bibitem{BolNabOymPau:06}
H.~Bolcskei, R.~Nabar, O.~Oyman, and A.~Paulraj, ``Capacity scaling laws in
  {MIMO} relay networks,'' \emph{Wireless Communications, IEEE Transactions
  on}, vol.~5, no.~6, pp. 1433--1444, June 2006.

\bibitem{ChiWinShi:C06}
M.~Chiani, M.~Z. Win, and H.~Shin, ``A general result on hypergeometric
  functions of matrix arguments and application to wireless {MIMO}
  communication,'' in \emph{Proc. IEEE International Conference on
  Next-Generation Wireless Systems (ICNEWS'06)}, Dhaka, Bangladesh, Jan. 2006,
  pp. 196--200, {\bf Invited Paper}.

\bibitem{JinMcKGaoCol:08}
S.~Jin, M.~McKay, X.~Gao, and I.~Collings, ``{MIMO} multichannel beamforming:
  {SER} and outage using new eigenvalue distributions of complex noncentral
  {Wishart} matrices,'' \emph{Communications, IEEE Transactions on}, vol.~56,
  no.~3, pp. 424--434, March 2008.

\bibitem{ShiWin:08}
H.~Shin and M.~Win, ``{MIMO} diversity in the presence of double scattering,''
  \emph{Information Theory, IEEE Transactions on}, vol.~54, no.~7, pp.
  2976--2996, July 2008.

\bibitem{KanKwaPraStu:08}
H.~Kang, J.~S. Kwak, T.~Pratt, and G.~Stuber, ``Analytical framework for
  optimal combining with arbitrary-power interferers and thermal noise,''
  \emph{Vehicular Technology, IEEE Transactions on}, vol.~57, no.~3, pp.
  1564--1575, May 2008.

\bibitem{McKZanColChi:J09}
M.~R. McKay, A.~Zanella, I.~B. Collings, and M.~Chiani, ``Error probability and
  {SINR} analysis of optimum combining in {R}ician fading,'' \emph{IEEE Trans.
  Commun.}, vol.~57, no.~3, pp. 676--687, Mar. 2009.

\bibitem{MarOlk:B79}
A.~W. Marshall and I.~Olkin, \emph{Inequalities: Theory of Majorization and Its
  Applications}.\hskip 1em plus 0.5em minus 0.4em\relax New York: Academic,
  1979.

\bibitem{Kha:66}
C.~G. Khatri, ``On certain distribution problems based on positive definite
  quadratic functions in normal vectors,'' \emph{Ann. Math. Statist.}, vol.~37,
  pp. 468--479, 1966.

\bibitem{Hay:66}
T.~Hayakawa, ``On the distribution of a quadratic form in a multivariate normal
  sample,'' \emph{The Inst. of Stat. Math.}, vol.~2, pp. 191--201, 1966.

\end{thebibliography}


%

\begin{figure}[p]
\psfrag{SNR (dB)}{\small SNR (dB)}
\psfrag{E[C]}{\small \hspace{-2cm}Ergodic mutual information  (bits/s/Hz)}
\centerline{\includegraphics[width=0.8\linewidth,draft=false]{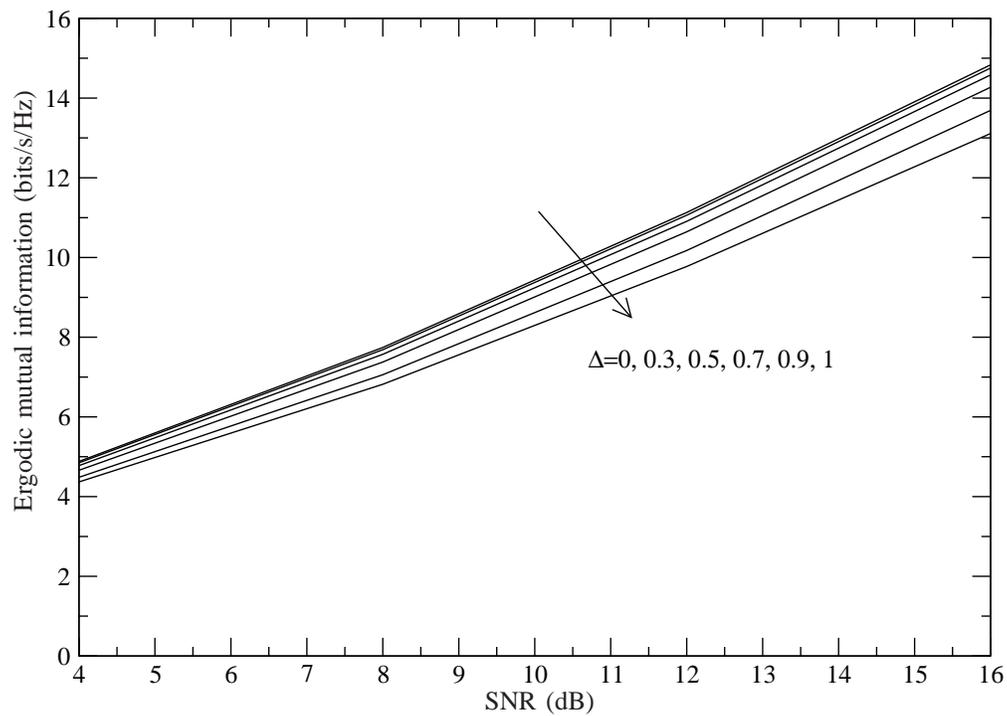}}
\caption{Ergodic mutual information for single-user MIMO systems as a function
of \ac{SNR} over Rayleigh uncorrelated fading with $\nt=6$,
$\nr=3$. Half of the antennas with power (normalized) $1+\Delta$,
the others with $1-\Delta$, i.e., with transmitted power levels
(normalized) equal to $\{1+\Delta, 1+\Delta, 1+\Delta, 1-\Delta,
1-\Delta, 1-\Delta \}$.} \label{fig:ECnt6nr3nq2}
\end{figure}

\newpage


\begin{figure}[p]
\psfrag{SNR}{\small SNR (dB)}
\psfrag{Cu}{\small \hspace{-1cm} $C_u$ (bits/s/Hz)}
\centerline{\includegraphics[width=0.8\linewidth,draft=false]{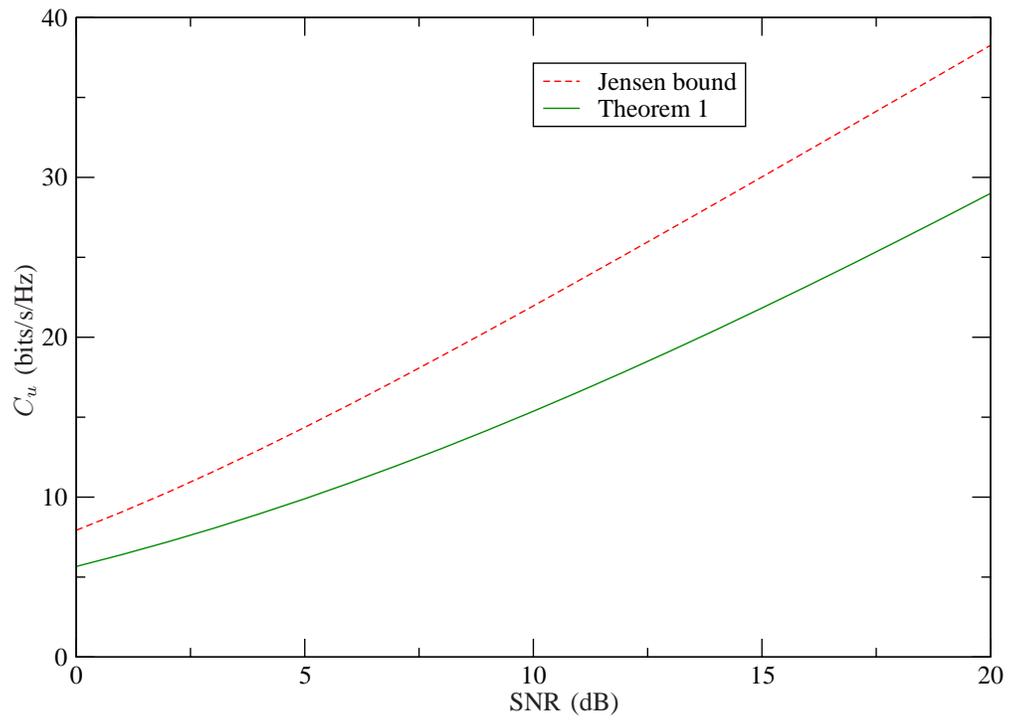}}
\caption{
\mynote{
Bound on the network capacity for MIMO relay networks. Source with $4$ antennas, $5$ relays with $2$ antennas each, power levels per relay proportional to $\{1,2,5,10,20\}$.} 
}
\label{fig:MIMOrelay}
\end{figure}

\newpage

\begin{figure}[p]
\psfrag{E[C]}{\small \hspace{-2cm}Ergodic mutual information  (bits/s/Hz)}
\psfrag{SIR}{\small SIR (dB)}
\centerline{\includegraphics[width=0.8\linewidth,draft=false]{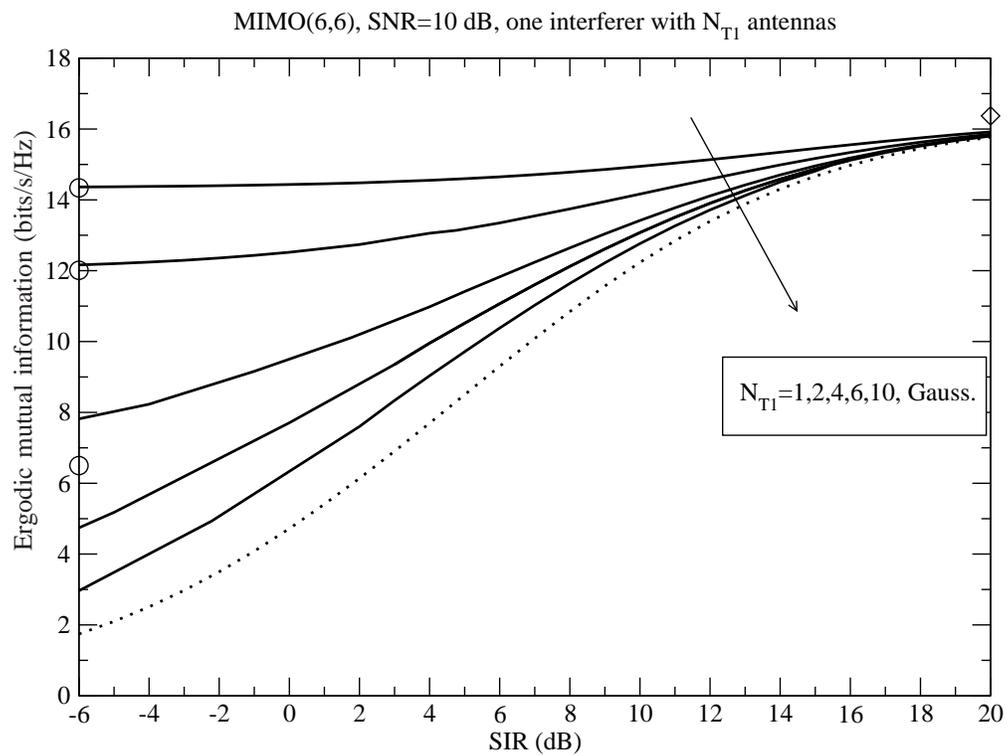}}
\caption{Ergodic mutual information for MIMO-(6,6) as a function of \ac{SIR} in
the presence of one MIMO cochannel interferer with
$\Nt_1=1,2,4,6,10$. The \ac{SNR} is set to 10 dB. The
Gaussian approximation of the interference is also shown. Diamond:
capacity of a single-user MIMO-$(6,6)$. Circles: capacity of a
single-user MIMO-$(6,6-\Nt_1)$ (only for $\Nt_1=1,2,4$).}
\label{figuraMIMOCCISNR10dbNR6.eps}
\end{figure}

\newpage

\begin{figure}[p]
\psfrag{E[C]}{\small \hspace{-2cm}Ergodic mutual information  (bits/s/Hz)}
\centerline{\includegraphics[width=0.8\linewidth,draft=false]{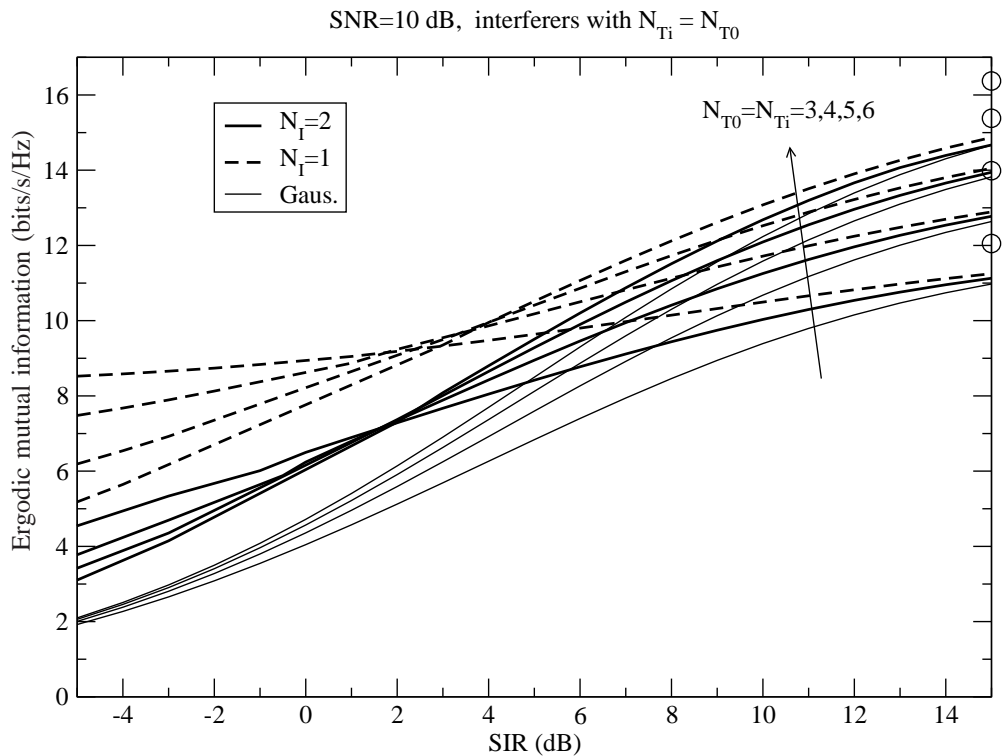}}
\caption{Ergodic mutual information as a function of the signal-to-total
interference ratio. MIMO system with $\Nr=6$ receiving antenna,
\textsf{SNR} = 10 dB. The Gaussian approximation of the interference is also shown. Scenario with one and two interferers, each
with the same number of transmitting antennas as the desired user.
Cases of $3,4,5$ and $6$ transmitting antennas. Circles: capacity
of single-user MIMO-$(\Nt_0,\Nr)$.}
\label{figuraMIMOnR6CCISNR10dbONEandTWOandGAUSINTERF.eps}
\end{figure}

\clearpage

\begin{biography}
{Marco Chiani}  (M'94--SM'02) was born in Rimini, Italy, in April 1964.  He received the Dr. Ing. degree (magna cum laude) in Electronic Engineering and the Ph.D. degree in Electronic and Computer Science from the University of Bologna in 1989 and 1993, respectively.  Dr. Chiani is a Full Professor at the II Engineering Faculty, University of Bologna, Italy, where he is the Chair in Telecommunication.  During the summer of 2001 he was a Visiting Scientist at AT\&T Research Laboratories in Middletown, NJ.  He is a frequent visitor at the Massachusetts Institute of Technology (MIT), where he presently holds a Research Affiliate appointment.

Dr. Chiani's research interests include wireless communication systems, MIMO systems, wireless multimedia, low density parity check codes (LDPCC) and UWB. He is leading the research unit of University of Bologna on cognitive radio and UWB (European project EUWB), on Joint Source and Channel Coding for wireless video (European projects Phoenix-FP6 and Optimix-FP7), and is a consultant to the European Space Agency (ESA-ESOC) for the design and evaluation of error correcting codes based on LDPCC for space CCSDS applications.  

Dr. Chiani has chaired, organized sessions and served on the Technical Program Committees at several IEEE International Conferences. In January 2006 he received the ICNEWS award "For Fundamental Contributions to the Theory and Practice of Wireless Communications". He was the recipient of the 2008 IEEE ComSoc Radio Communications Committee Outstanding Service Award.

He is the past chair (2002-2004) of the Radio Communications Committee of the IEEE Communication Society and past Editor of Wireless Communication (2000-2007) for the {\scshape IEEE Transactions on Communications}.   
\end{biography}

\begin{biography}
{\bf Moe Z. Win} (S'85-M'87-SM'97-F'04)
received both the Ph.D.\ in Electrical Engineering and M.S.\ in Applied Mathematics
as a Presidential Fellow at the University of Southern California (USC) in 1998.
He received an M.S.\ in Electrical Engineering from USC in 1989, and a B.S.\ ({\em magna cum laude})
in Electrical Engineering from Texas A\&M University in 1987.

Dr.\ Win is an Associate Professor at the Massachusetts Institute of Technology (MIT).
Prior to joining MIT, he was at AT\&T Research
Laboratories for five years and at the Jet Propulsion Laboratory for seven years.
His research encompasses developing fundamental theories, designing algorithms, and
conducting experimentation for a broad range of real-world problems.
His current research topics include location-aware networks,
time-varying channels, multiple antenna systems, ultra-wide bandwidth
systems, optical transmission systems, and space communications systems.

Professor Win is an IEEE Distinguished Lecturer and
        elected Fellow of the IEEE, cited for ``contributions to wideband wireless transmission.''
He was honored with
        the IEEE Eric E. Sumner Award (2006), an IEEE Technical Field Award for
        ``pioneering contributions to ultra-wide band communications science and technology.''
Together with students and colleagues, his papers have received several awards including
        the IEEE Communications Society's Guglielmo Marconi Best Paper Award (2008)
    and the IEEE Antennas and Propagation Society's Sergei A. Schelkunoff Transactions Prize Paper Award (2003).
His other recognitions include
        the Laurea Honoris Causa from the University of Ferrara, Italy (2008),
        the Technical Recognition Award of the IEEE ComSoc Radio Communications Committee (2008),
        Wireless Educator of the Year Award (2007),
        the Fulbright Foundation Senior Scholar Lecturing and Research Fellowship (2004),
        the U.S. Presidential Early Career Award for Scientists and Engineers (2004),
        the AIAA Young Aerospace Engineer of the Year (2004),
    and the Office of Naval Research Young Investigator Award (2003).

Professor Win has been actively involved in organizing and chairing
a number of international conferences. He served as
    the Technical Program Chair for
        the IEEE Wireless Communications and Networking Conference in 2009,
        the IEEE Conference on Ultra Wideband in 2006,
        the IEEE Communication Theory Symposia of ICC-2004 and Globecom-2000,
        and
        the IEEE Conference on Ultra Wideband Systems and Technologies in 2002;
    Technical Program Vice-Chair for
        the IEEE International Conference on Communications in 2002; and
    the Tutorial Chair for
        ICC-2009 and
        the IEEE Semiannual International Vehicular Technology Conference in Fall 2001.
He was
    the chair (2004-2006) and secretary (2002-2004) for
        the Radio Communications Committee of the IEEE Communications Society.
Dr.\ Win is currently
    an Editor for {\scshape IEEE Transactions on Wireless Communications.}
He served as
    Area Editor for {\em Modulation and Signal Design} (2003-2006),
    Editor for {\em Wideband Wireless and Diversity} (2003-2006), and
    Editor for {\em Equalization and Diversity} (1998-2003),
        all for the {\scshape IEEE Transactions on Communications}.
He was Guest-Editor
        for the
        {\scshape Proceedings of the IEEE}
        (Special Issue on UWB Technology \& Emerging Applications) in 2009 and
        {\scshape IEEE Journal on Selected Areas in Communications}
        (Special Issue on Ultra\thinspace-Wideband Radio in Multiaccess
        Wireless Communications) in 2002.
\end{biography}

\begin{biography}
{Hyundong Shin} (S'01--M'04) received the B.S. degree in
Electronics Engineering from Kyung Hee University, Korea, in 1999,
and the M.S. and Ph.D. degrees in Electrical Engineering from
Seoul National University, Seoul, Korea, in 2001 and 2004,
respectively.

From September 2004 to February 2006, Dr.\ Shin was a Postdoctoral
Associate at the Laboratory for Information and Decision Systems
(LIDS), Massachusetts Institute of Technology (MIT), Cambridge,
MA, USA. In March 2006, he joined the faculty of the School of
Electronics and Information, Kyung Hee University, Korea, where he
is now an Assistant Professor at the Department of Electronics and
Radio Engineering. His research interests include wireless
communications, information and coding theory,
cooperative/collaborative communications, and multiple-antenna
wireless communication systems and networks.

Professor Shin served as a member of the Technical Program
Committee in
    the IEEE International Conference on Communications (2006, 2009),
    the IEEE International Conference on Ultra Wideband (2006),
    the IEEE Global Communications Conference (2009, 2010),
    the IEEE Vehicular Technology Conference (2009 Fall, 2010 Spring),
    and
    the IEEE International Symposium on Personal, Indoor and Mobile Communications (2009).
He served as a Technical Program co-chair for
    the IEEE Wireless Communications and Networking Conference PHY Track (2009).
Dr.\ Shin is currently
    an Editor for {\scshape IEEE Transactions on Wireless Communications}.
He was a Guest Editor for
    the 2008 {\scshape EURASIP Journal on Advances in Signal Processing} (Special Issue on Wireless
    Cooperative Networks).

Professor Shin received
    the IEEE Communications Society's Guglielmo Marconi Prize Paper Award (2008)
    and
    the IEEE Vehicular Technology Conference Best Paper Award
    (2008 Spring).

\end{biography}

\end{document}